\documentclass[10pt,a4paper]{article}

\usepackage[margin=25mm]{geometry}
\usepackage[british]{babel}
\usepackage[T1]{fontenc}
\usepackage[utf8]{inputenc}
\usepackage{lmodern}
\usepackage{amsmath,amssymb,amsthm,mathtools}
\usepackage{braket}
\usepackage{microtype}
\usepackage{graphicx,xcolor}
\usepackage{tikz}
\usepackage{authblk}
\usepackage{booktabs}

\usetikzlibrary{arrows.meta,calc,decorations.pathreplacing}
\usepackage{amsmath,tikz,graphicx}
\usepackage{enumitem}
\setlist{topsep=4pt,itemsep=2pt,parsep=0pt}
\usepackage{cite}
\usepackage{hyperref}

\usepackage{bookmark}
\usepackage{tocbasic}

\usepackage{float}
\usepackage{algorithm}

\DeclareTOCStyleEntries[
  dynnumwidth,
  numwidth=0pt,
  numsep=0.6em
]{tocline}{section,subsection}

\DeclareTOCStyleEntry[
  dynindent,
  indentfollows=section
]{tocline}{subsection}

\newtheorem{theorem}{Theorem}
\newtheorem{proposition}{Proposition}
\newtheorem{corollary}{Corollary}
\newtheorem{lemma}{Lemma}
\newcommand{\restatedtitle}{}

\newtheorem*{restatedresult}{\restatedtitle}
\newenvironment{namedresult}[1]{\renewcommand{\restatedtitle}{#1}\begin{restatedresult}}{\end{restatedresult}}
\theoremstyle{definition}
\newtheorem{definition}{Definition}
\theoremstyle{remark}

\newcommand{\HT}{\mathrm{HT}}
\newcommand{\HTp}{\mathrm{HT}^{+}}
\newcommand{\Setup}{\texttt{Setup}}
\newcommand{\Update}{\texttt{Update}}
\newcommand{\Checkk}{\texttt{Check}}
\newcommand{\Shift}{\textsc{Shift}}
\newcommand{\refX}{\mathrm{ref}_{\mathcal{X}}}

\providecommand{\QFF}{}\renewcommand{\QFF}{\hyperref[thm:qff]{\textsc{QFF}}}
\providecommand{\KMOR}{}\renewcommand{\KMOR}{\hyperref[alg:kmorqff]{\textsc{KMORQFF}}}
\providecommand{\IIS}{}\renewcommand{\IIS}{\hyperref[alg:iis]{\textsc{InterpolatedIncrementalSearch}}}
\providecommand{\BPCS}{}\renewcommand{\BPCS}{\hyperref[alg:bpcs]{\textsc{BoundedCandidateSearch}}}
\providecommand{\IAGJK}{}\renewcommand{\IAGJK}{\hyperref[alg:agjk]{\textsc{InterpolatedAGJK}}}

\DeclarePairedDelimiter{\ceil}{\lceil}{\rceil}

\allowdisplaybreaks[2]

\title{Quantum-classical algorithm interpolation for spatial search}

\author[1]{Diogo Ramos}
\author[1,2]{Duarte Magano\thanks{%
  \href{mailto:duarte.magano@fc.up.pt}
       {\nolinkurl{duarte.magano@fc.up.pt}}}}

\affil[1]{%
  Centro de Física das Universidades do Minho e do Porto, LaPMET,
  Departamento de Física e Astronomia, Faculdade de Ciências,
  Universidade do Porto,
  Rua do Campo Alegre, 4169-007 Porto, Portugal}

\affil[2]{21strategies GmbH, Hallbergmoos, Germany}

\date{}
\begin{document}
\maketitle
\begin{abstract}
Quantum-walk algorithms can quadratically reduce the dependence of spatial search on the classical hitting time $\mathrm{HT}$. Realising this speedup using standard quantum-walk algorithms, however, requires deep coherent circuits, limiting their applicability on near-term and early fault-tolerant quantum devices. We derive a parameterised family of hybrid quantum-classical spatial search algorithms that continuously interpolate between the classical and fully coherent quantum regimes, trading off circuit depth for total query complexity. For any $\gamma\in[0,1]$, there is an algorithm that uses circuits of depth $\widetilde{O}(\mathrm{HT}^{\gamma/2})$ and finds a marked vertex with constant success probability at total cost $\widetilde{O}(\mathrm{HT}^{1-\gamma}\mathsf{S} + \mathrm{HT}^{1-\gamma/2}(\mathsf{U} + \mathsf{C}))$, where $\mathsf{S}$, $\mathsf{U}$, and $\mathsf{C}$ denote the setup, update, and checking costs, respectively. We also derive a worst-case optimal continuous depth-query trade-off for fixed-point amplitude amplification at constant failure probability, providing a potential building block for other hybrid designs.
\end{abstract}

\setcounter{tocdepth}{2}
\tableofcontents
\medskip

\section{Introduction}\label{ch:intro}

Quantum computers promise efficient solutions to problems that are intractable on classical computers.
Translating these asymptotic advantages into practical gains, however, remains severely constrained by hardware limitations. Current noisy intermediate-scale quantum \cite{preskill2018quantum} and early fault-tolerant  \cite{katabarwa2024early} devices support only limited circuit depths, with noise accumulating as circuit length increases.
To accommodate these hardware constraints, many near-term quantum algorithms adopt a hybrid architecture in which short, specialised coherent quantum subroutines are embedded within a classical optimisation loop \cite{cerezo2021variational}. Prominent examples include the Variational Quantum Eigensolver \cite{peruzzo2014variational} and the Quantum Approximate Optimization Algorithm \cite{farhi2014qaoa,mcclean2016theory}.

Seeking a rigorous characterisation of the hybrid computing paradigm, Sun and Zheng introduced the hybrid query/decision-tree model \cite{sun2019hybrid}.
In this model, Boolean functions are computed by sequentially executing quantum circuits, each of which performs at most $D$ queries.
The key question is: given a strict cap $D$ on sequential quantum depth, how should the computational effort be divided between quantum and classical resources to minimise the time to solution?
The introduction of the hybrid query model initiated a programme to rigorously quantify  the trade-off between quantum coherence and query complexity.

As a first example, consider unstructured search.
Grover's algorithm finds a marked element in a search space of size $N$ with $O(\sqrt{N})$ queries \cite{grover1996fast}.
Given a depth limit $D$, the natural strategy is to partition the search domain into $N/D^2$ blocks and apply Grover's algorithm to each block.
Sun and Zheng showed that this strategy is optimal, establishing a hybrid query complexity of $\Theta(N/D)$ for unstructured search \cite{sun2019hybrid}.

Another early step in the hybrid complexity programme was taken by Wang, Higgott, and Brierley~\cite{wang2019accelerated}.
Consider a randomised algorithm $\mathcal{A}$ that outputs \texttt{YES} with probability $p$ (and \texttt{NO} with probability $1-p$).
Estimating $p$ to additive error $\epsilon$ by repeated sampling requires $\Theta(1/\epsilon^2)$ samples.
Amplitude estimation \cite{brassard2002amplification} provides such an estimate using a single coherent quantum circuit that makes $O(1/\epsilon)$ queries to $\mathcal{A}$.
Wang et al.~\cite{wang2019accelerated} proposed a hybrid algorithm that, for any $\gamma \in [0,1]$, uses circuits of depth at most $D(\gamma)$ and performs a total of $T(\gamma)$ queries, where
\begin{equation}
    T(\gamma)=O\!\left(\frac{1}{\epsilon^{2-\gamma}}\right)
    \text{ and }
    D(\gamma)=O\!\left(\frac{1}{\epsilon^{\gamma}}\right).
\end{equation}
This construction provides a continuous interpolation between the classical ($\gamma = 0$) and quantum ($\gamma = 1$) limits.
We illustrate the concept in Figure~\ref{fig:interpolation_philosophy}.
The result was later proved rigorously in Ref.~\cite{giurgica2022qjournal}, with alternative derivations given in Refs.~\cite{magano2022simplifying,Rall2023amplitudeestimation,vu2025lowdepth}.
Ambainis, Gily\'en, and Kokainis~\cite{ambainis2026tight} recently established its optimality.

\begin{figure*}
  \centering
  \includegraphics[width=\textwidth]{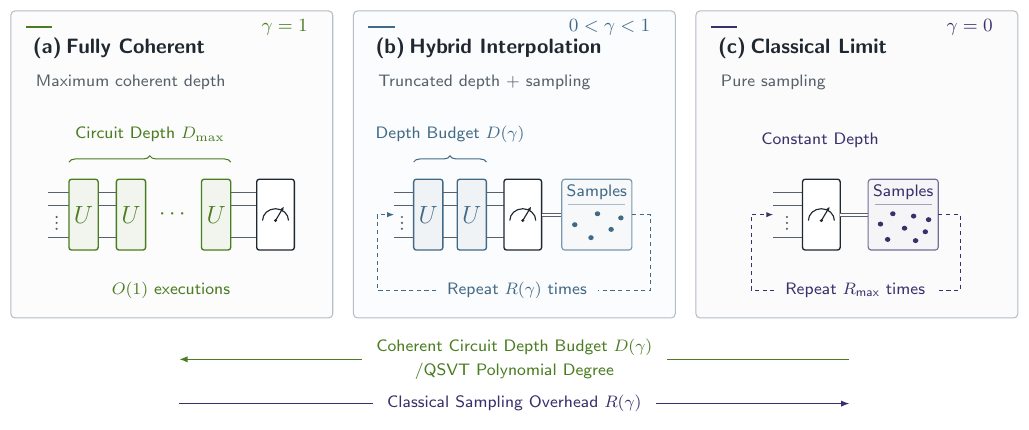}
  \caption{The framework of quantum-classical interpolation. The parameter $\gamma \in [0, 1]$ smoothly trades coherent circuit depth $D(\gamma)$ for classical repetitions $R(\gamma)$, continuously connecting deep unitary evolution ($\gamma = 1$) to classical sampling ($\gamma = 0$).}
  \label{fig:interpolation_philosophy}
\end{figure*}

Beyond these two problems, relatively few continuous quantum-classical interpolations are known.
Examples include phase estimation (with an upper bound related to amplitude estimation and a lower bound from Ref.~\cite{ambainis2026tight}),
AND-OR trees (with an upper bound from Ref.~\cite{murca2024making} and a lower bound for two-level trees from Ref.~\cite{ambainis2026tight}), and element distinctness (with an upper bound from Ref.~\cite{ambainis2026tight}).
In this work, we develop such an interpolation for \emph{spatial search}.

Grover's algorithm assumes an unstructured search space, which does not accurately represent many search problems.
For instance, molecular configurations may be connected by transitions that transform one into another, while web pages are connected by hyperlinks.
In such settings, the data form a structured space, and accessing one item from another depends on the connections between them.
This abstraction is naturally modelled as a graph.
Let $G=(X,E)$ be a connected undirected graph on $N:=|X|$ vertices and let $M\subseteq X$ be a non-empty set of marked vertices. Suppose that a marked vertex can be recognised only when it is visited and that movement along the graph is restricted to its edges.
The task, known as spatial search, is to output an element of $M$ while respecting these locality constraints. Unstructured search is recovered when $G$ is the complete graph, in which every vertex is directly accessible from every other vertex.

Classically, such settings are often tackled using random-walk algorithms because a walker on a connected graph eventually visits every vertex. Let $P$ denote the transition matrix of a random walk on $G$. A walker samples a starting vertex, checks whether it is marked, and, if not, moves to one of its neighbours according to $P$; this process is repeated until a marked vertex is found. The expected number of steps required to reach $M$ is the hitting time $\HT(P,M)$, which depends on the topology of the underlying graph.

Random-walk algorithms involve three distinct operations, whose costs are usefully accounted for separately \cite{magniez2011search,kmor,ambainisqff,apers2021unified}.
Preparing the initial sample incurs a cost $\mathsf S$, advancing the walk by one step, a cost $\mathsf U$, and determining whether the current vertex is marked, a cost $\mathsf C$. These notions are rigorously defined in Section~\ref{sec:costmodel}. In these terms, the cost of finding a
marked vertex by classical random walk search is of order
\begin{equation}\label{eq:intro-classical}
\mathsf S+\HT(P,M)\,(\mathsf U+\mathsf C).
\end{equation}

Advances in the theory of quantum walks and quantum search have established that replacing a classical random walk with coherent quantum-walk dynamics can yield a quadratic reduction in the number of steps required \cite{Sze04a, magniez2012hitting,ambainisqff,kmor,apers2021unified,piddock2019,apers2022ctqw}.
This extends Grover's quadratic speedup to spatial search, reducing the dependence on the classical hitting time from $\HT(P,M)$ to $\sqrt{\HT(P,M)}$, up to logarithmic factors.

Using the quantum-walk schemes of Refs.~\cite{kmor,ambainisqff}, we develop hybrid quantum-classical algorithms for spatial search that interpolate continuously between the classical hitting-time dependence $\HT$ and the quantum scaling $\sqrt{\HT}$.
For $\gamma\in[0,1]$, we obtain algorithms with quantum-walk depth
\begin{equation}
    \widetilde O\left(\HT^{\gamma/2}\right)
\end{equation}
that find a marked vertex with constant success probability and hybrid complexity
\begin{equation}\label{eq:intro-interpolation}
    \widetilde O\!\left(
        \HT^{1-\gamma}\cdot\mathsf S
        \;+\;
        \HT^{1-\gamma/2}\cdot(\mathsf U+\mathsf C)
    \right),
\end{equation}
for any number of marked vertices. The dependence on the checking cost is optimal (in terms of the hitting time) in the worst case over graphs, up to logarithmic factors, since spatial search on the complete graph recovers unstructured search and its corresponding hybrid query lower bound.

Our interpolation scheme preserves the simplicity of the original constructions while requiring a refined analysis of their behaviour in the limited-depth regime.  To this end, we recast the algorithm of Ref.~\cite{kmor} in terms of quantum fast-forwarding~\cite{ApersSarlette2019}, showing that both search algorithms of Refs.~\cite{kmor,ambainisqff} can be formulated in terms of the same primitive.

Crucially, by separating setup, update, and checking costs, we reveal how the cost of initial-state preparation affects the utility of shallow quantum circuits, exposing a structural bottleneck absent from standard query-complexity models. In these algorithms, reducing the coherent depth increases the number of independent runs, each requiring a fresh invocation of the initial-state preparation routine. The resulting trade-off therefore depends not only on the available depth but also on the relative costs of preparing the initial state and advancing the walk by one step.

As a preliminary step, we begin with a continuous interpolation for fixed-point amplitude amplification \cite{PhysRevLett.113.210501} that, to the best of our knowledge, has not previously been formalised explicitly in the literature.

\newpage
\section{Interpolated fixed-point amplitude amplification}
\label{sec:interp-aa}

Before analysing spatial search, we first consider the simpler problem of trading coherent depth for sampling overhead in fixed-point amplitude amplification. Here, preparing the initial state requires only a single query. Comparing this setting with the quantum walk cost model, in which initial-state preparation is accounted for separately, clarifies how locality constraints affect the limited-depth regime.

\begin{figure*}[t]
  \centering
  \includegraphics[width=\textwidth]{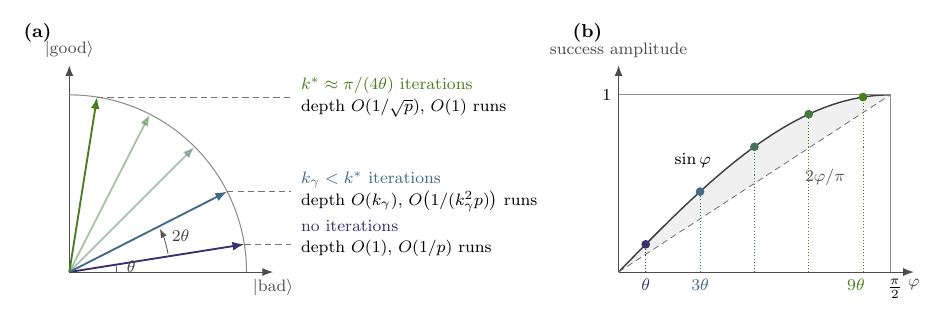}
  \caption{Interpolated amplitude amplification in its original formulation, with $\sin\theta=\sqrt p$.
  (a)~Each Grover iteration rotates the state by $2\theta$ towards $\ket{\mathrm{good}}$. Stopping after $k_\gamma<k^*$ iterations reduces the depth, and the resulting reduction in success probability is compensated for by repeating the procedure.
  (b)~Success amplitude $\sin\varphi$ after $k$ iterations, where $\varphi=(2k+1)\theta$. For $\varphi\le\pi/2$, concavity ensures the amplitude grows at least linearly with $k$ and the number of runs decreases quadratically with the depth.}
  \label{fig:aa-rotation}
\end{figure*}

Consider a unitary $A$ acting on $m$ qubits such that
\begin{equation}\label{eq:aa-split}
A\ket{0^m}=\sqrt p\,\ket{\mathrm{good}}+\sqrt{1-p}\,\ket{\mathrm{bad}},
\end{equation}
where $\ket{\mathrm{good}}$ and $\ket{\mathrm{bad}}$ are normalised states in orthogonal ``good'' and ``bad'' subspaces. Write $\sqrt p=\sin\theta$. Let $\Pi := \ket{0^m}\bra{0^m}$ and let $\Pi'$ denote the projector onto the good subspace. Assume access to $A$, $A^\dagger$, and the reflections $R_\Pi := I - 2\Pi$ and $R_{\Pi'} := I - 2\Pi'$.
The task is to prepare a state in the good subspace with failure probability at most $\zeta\in(0,1/2)$. The query depth $D$ is the maximum number of calls to $A$, $A^\dagger$, and $R_{\Pi'}$ in a single coherent run, whereas the total query complexity counts these calls across all runs.

In the original formulation of amplitude amplification~\cite{brassard2002amplification}, each application of the Grover operator $G:=-AR_{\Pi}A^\dagger R_{\Pi'}$ rotates the state by $2\theta$ towards $\ket{\mathrm{good}}$, so that $G^kA\ket{0^m}$ has amplitude $\sin((2k+1)\theta)$ on the state $\ket{\mathrm{good}}$. If $p$ is known, the interpolation is immediate, as illustrated in Figure~\ref{fig:aa-rotation}. Instead of the $k^*\approx\pi/(4\theta)$ iterations of the full-depth algorithm, we perform only $k_\gamma=\Theta(\theta^{-\gamma})$ iterations, while keeping $(2k_\gamma+1)\theta\le\pi/2$. Then by Jordan's inequality each run succeeds with probability at least $\tfrac{4}{\pi^2}(2k_\gamma+1)^2p=\Omega(p^{1-\gamma})$ at depth $O(p^{-\gamma/2})$. Repeating the procedure $O(p^{\gamma-1})$ times gives constant success probability with a total of $O(p^{\gamma/2-1})$ queries.

Without knowledge of $p$, a fixed number of iterations may overshoot the good subspace, a phenomenon commonly known as the \emph{soufflé} problem.
This is addressed by fixed-point amplitude amplification~\cite{PhysRevLett.113.210501}.
Suppose now that the initial success probability $p\in(0,1)$ is unknown but that a lower bound $0<\delta\le\sqrt p$ is known.
For any real odd polynomial \(P\) of degree \(d\) bounded by one on \([-1,1]\), we can implement the quantum singular value transformation (QSVT) \(\sqrt p \rightarrow P(\sqrt p)\) \cite{Gily_n_2019}. The corresponding circuit uses \(O(d)\) calls to $A$, $A^{\dagger}$, controlled-$R_\Pi$, and controlled-$R_{\Pi'}$, as well as single-qubit rotation gates.
Since the compression
\begin{equation}\label{eq:aa-compression}
\Pi' A\Pi=\sqrt p\,\ket{\mathrm{good}}\bra{0^m}\,,
\end{equation}
encodes the success amplitude $\sqrt p$ in a block of $A$, the challenge is to find an odd polynomial for which the amplitude $P(\sqrt p)$ is close to $1$ for $\sqrt p\in[\delta,1]$. The following standard sign-approximation result provides such a polynomial.
For a constant approximation error, we recover the familiar $O(1/\delta)$ query depth for constant success probability.
\newpage
\begin{lemma}[{\cite[Lemma~3]{lin2020near}}]
\label{l:sign}
Let $\delta,\epsilon\in(0,1)$. There is an odd real polynomial $P_{\Re}$ of degree $O\!\left(1+\delta^{-1}\log(1/\epsilon)\right)$ satisfying
\[
\begin{aligned}
|P_{\Re}(x)|&\le1
&&\text{for all }x\in[-1,1],\\
|P_{\Re}(x)-\operatorname{sign}(x)|&\le\epsilon
&&\text{for all }x\in[-1,-\delta]\cup[\delta,1].
\end{aligned}
\]
\end{lemma}
\noindent

To derive an interpolating regime, we follow the approach of Refs.~\cite{magano2022simplifying,vu2025lowdepth}, which analyse amplitude estimation through the lens of QSVT.
In this view, query depth corresponds to the degree of the polynomial transformation.
A deeper circuit yields a more accurate polynomial approximation to the target function, thereby increasing the statistical information extractable from each measurement and reducing the number of samples required.
For fixed-point amplitude amplification, shorter circuits implement less accurate approximations to the sign function, corresponding to larger values of the approximation error $\epsilon$.
The following theorem makes this idea explicit.

\begin{figure}[t]
\centering
\includegraphics[width=0.72\textwidth]{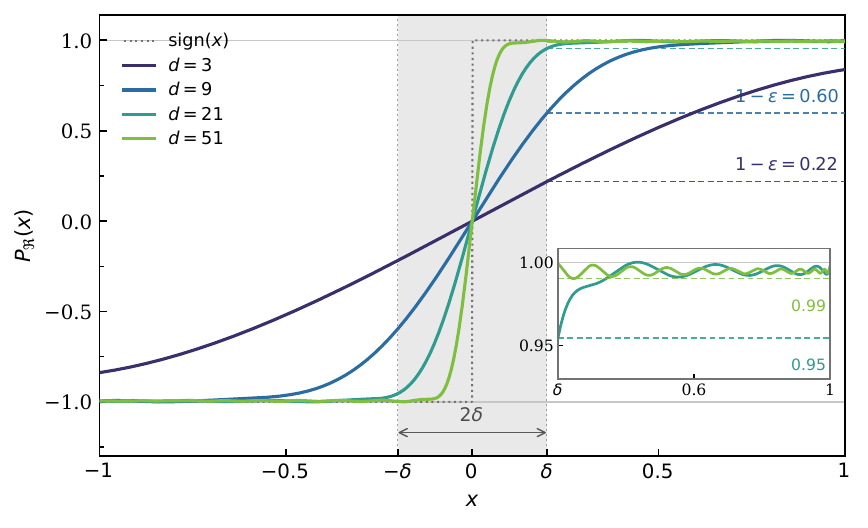}
\caption{Illustration of the degree--accuracy trade-off for odd polynomials $P_{\Re}$ approximating $\operatorname{sign}(x)$ at fixed resolution $\delta=0.2$. The displayed errors on $[\delta,1]$ range from $\epsilon\approx0.78$ at $d=3$ to $\epsilon\approx0.01$ at $d=51$. The horizontal dashed lines indicate the corresponding amplitude levels $1-\epsilon$. A lower bound of $1-\epsilon$ throughout $[\delta,1]$ gives success probability at least $(1-\epsilon)^2$ per attempt, relating the polynomial approximation directly to the required repetitions.}
\label{fig:sign-approx}
\end{figure}

\begin{theorem}
\label{thm:interp-aa}
Given a known lower bound $0<\delta\le\sqrt p$ and access to $A$, $A^\dagger$, and $R_{\Pi'}$, for any target failure probability $\zeta\in(0,1/2)$ and any $\gamma\in[0,1]$, there is a hybrid quantum-classical algorithm with query depth
\begin{equation}\label{eq:aa-depth}
D(\gamma)=O\!\left(\frac{1}{\delta^\gamma}\right)
\end{equation}
that outputs a state in the good subspace with probability at least $1-\zeta$ and has total query complexity
\begin{equation}\label{eq:aa-total}
T(\gamma)=O\!\left(\frac{\log(1/\zeta)}{\delta^{2-\gamma}}\right).
\end{equation}
\end{theorem}

\begin{proof}
Fix some $\gamma\in[0,1]$. We invoke Lemma~\ref{l:sign} to block-encode a polynomial approximation of the sign function with
\begin{equation}\label{eq:aa-eps}
\epsilon=1-\tfrac12\delta^{\,1-\gamma}\in[1/2,1).
\end{equation}
See Figure~\ref{fig:sign-approx} for an illustration of the underlying transformation.
Since $x\le-\log(1-x)\le2x$ for $0<x\le1/2$, the polynomial $P_{\Re}$ has degree
\begin{equation}\label{eq:aa-depth-derivation}
D(\gamma)
=O\!\left(1+\frac1\delta\log\frac1{1-\tfrac12\delta^{\,1-\gamma}}\right)
=O\!\left(\delta^{-\gamma}\right),
\end{equation}
matching the query depth in \eqref{eq:aa-depth}.
For every $\sqrt p\ge\delta$, the polynomial guarantee gives a success amplitude
\begin{equation}\label{eq:aa-Pdelta}
P_{\Re}(\sqrt p)\ge1-\epsilon=\tfrac12\delta^{\,1-\gamma}.
\end{equation}
As all runs are independent, requiring the probability of failing $R(\gamma)$ consecutive times to be at most $\zeta$ gives
\begin{equation}\label{eq:aa-R-derivation}
R(\gamma)\ge\frac{\log(1/\zeta)}{-\log(1-\tfrac14\delta^{\,2(1-\gamma)})}.
\end{equation}
By the same bound on the logarithm, it is then sufficient to take
\begin{equation}\label{eq:aa-R}
R(\gamma)=\left\lceil\frac{4\log(1/\zeta)}{\delta^{\,2(1-\gamma)}}\right\rceil,
\end{equation}
and the total number of queries is therefore
\begin{equation}\label{eq:aa-total-derivation}
T(\gamma)=O\!\left(R(\gamma)D(\gamma)\right)
=O\!\left(\frac{\log(1/\zeta)}{\delta^{\,2-\gamma}}\right).\qedhere
\end{equation}
\end{proof}
At $\gamma=0$, the constant-depth circuit requires $O(\log(1/\zeta)/\delta^2)$ independent repetitions, matching the Bernoulli sampling bound for detecting an event of probability $\delta^2$. At $\gamma=1$, each attempt has depth $O(1/\delta)$ and succeeds with constant probability, requiring only $O(\log(1/\zeta))$ repetitions. Taking $\zeta$ to be a fixed constant recovers the usual sampling and coherent query scalings.
Moreover, the upper bound
\begin{equation}
T(\gamma)D(\gamma)=O(\log(1/\zeta)/\delta^{2})
\end{equation}
is independent of $\gamma$, so moving along the family only shifts effort from longer coherent runs to more repetitions. This trade-off is optimal in the worst case at constant failure probability. Indeed, unstructured search is recovered by taking $A=H^{\otimes m}$ and $\delta=1/\sqrt N$. In this case, finding a marked element with runs of depth $D$ provably requires $\Omega(\sqrt N+N/D)$ queries~\cite{sun2019hybrid,hamoudi2024nisq,ambainis2026tight}, matching Theorem~\ref{thm:interp-aa}.

The same picture does not necessarily carry over to spatial search. In amplitude amplification, the dynamics decomposes as a rotation in a two-dimensional subspace, so a truncated run immediately retains a success amplitude that grows linearly with its depth. A quantum walk, by contrast, evolves over its entire spectrum. Moreover, unlike a classical random walk, which may check the current vertex after every step and continue if it is unmarked, a quantum walk must preserve coherence throughout each run.  Establishing whether a truncated walk of depth $D$ still succeeds with probability of order $D^2/\HT$ therefore requires a separate analysis and constitutes the central challenge addressed in this work. Finally, just as Grover's algorithm begins each run by preparing the uniform superposition with $H^{\otimes m}$ before applying the Grover iterate, every run of a walk-based search begins with initial-state preparation. In the quantum-walk cost model, however, this preparation cost is accounted for explicitly.

\section{Main results}\label{sec:results}
The following results concern interpolations of quantum walk search, which connect the classical random walk step scalings at $\gamma=0$ with the fully coherent quantum bounds at $\gamma=1$. They are built on two different quantum walk search algorithms, leading to two families with complementary assumptions and guarantees.

The first approach builds upon the algorithm of Krovi, Magniez, Ozols and Roland~\cite{kmor}, which we henceforth refer to as the KMOR algorithm. In its original formulation, it starts from an unmarked component of the graph's vertices and projects onto a state with substantial overlap with the marked elements, mirroring the two-dimensional rotation in amplitude amplification. By recasting this projection via quantum fast-forwarding~\cite{ApersSarlette2019}, which simulates $t$ steps of the walk in $O(\sqrt{t})$ coherent queries, we can precisely quantify how truncating the evolution reduces overlap with the target state, and how many classical repetitions are needed to compensate for the diminished success probability.

\begin{theorem}\label{thm:cq_interp}
	Let $P$ be an ergodic, reversible Markov chain on $G = (X,E)$ and let
	$M \subseteq X$ contain a single marked vertex with probability of measuring it from the stationary distribution $p_M \leq \tfrac{1}{2}$. If the values of $p_M$ and
	$\HT:=\HT(P,M)$ are known, for any $\gamma \in [0,1]$, there exists
	a quantum walk-based algorithm with quantum-walk depth
	\begin{equation}\label{summary:eq:Tgamma}
		D(\gamma) \;=\;O\left(\HT^{\gamma/2}\right)
	\end{equation}
	which finds a marked vertex with constant success probability
	and hybrid complexity of order at most
	\begin{equation}\label{summary:eq:complexity_gamma}
		\HT^{1-\gamma} \cdot \mathsf{S}
		\;+\; \HT^{1-\gamma/2} \cdot (\mathsf{U} + \mathsf{C}).
	\end{equation}
\end{theorem}

At $\gamma=1$, Theorem~\ref{thm:cq_interp} achieves the full quadratic speedup using $O(1)$ runs of depth $O(\sqrt{\HT})$, whereas at $\gamma=0$, it reduces to constant-depth runs requiring $O(\HT)$ total walk steps. For intermediate $\gamma$, each depth-$D$ run succeeds with probability $\Omega(D^2/\HT)$.
In terms of the checking-query cost, i.e., the dependence on $\mathsf C$, this trade-off is optimal in the worst case over graphs, because spatial search on the complete graph recovers unstructured search. Since each run requires a fresh initialisation, the sampling overhead directly transfers to the setup cost, making interpolation advantageous only when $\mathsf{S}$ is cheap compared to performing a walk step.

For several marked vertices, the KMOR algorithm is governed by a quantity which may be much larger than the hitting time \cite{ambainisqff}, which motivates our second family.

\begin{theorem}\label{thm:interp-qff}
Let $P$ be an ergodic, reversible Markov chain on $G=(X,E)$,
let $M\subseteq X$ be non-empty, and suppose that
$\HT\ge\HT(P,M)$ is known.
For every $\gamma\in[0,1]$, there exists a quantum walk-based
algorithm that finds a marked vertex with constant success
probability, with maximum quantum-walk depth
\begin{equation}\label{eq:interp-qff-depth}
D(\gamma)=O\!\left(
\HT^{\gamma/2}
\sqrt{(1-\gamma)\log\HT+\log\log\HT}
\right)
\end{equation}
and hybrid complexity of order at most
\begin{equation}\label{eq:interp-qff-complexity}
(1+\gamma\log\HT)\HT^{1-\gamma}\mathsf S
+
(1+\gamma\log\HT)\HT^{1-\gamma/2}
\sqrt{(1-\gamma)\log\HT+\log\log\HT}
\,(\mathsf U+\mathsf C).
\end{equation}
\end{theorem}

The second approach builds on the work of Ambainis, Gilyén, Jeffery and Kokainis~\cite{ambainisqff} (hereafter AGJK), which also relies on quantum fast-forwarding but, instead of projecting onto a fixed target state, runs the walk for a randomly chosen number of steps and with randomly chosen parameters. With sufficient probability, some of these choices catch the walk while it visits a marked vertex, regardless of how many marked vertices there are or where they lie. In this case, truncating the length of the walk relates to showing that a classical walk run for fewer steps still reaches the marked set with a proportionally lower probability. This generality comes at the cost of logarithmic factors, so Theorem~\ref{thm:cq_interp} remains preferable for a single marked vertex with known $p_M$.

Both results extend to settings with weaker prior knowledge assumptions at the cost of at most polylogarithmic factors, as shown in Appendix~\ref{variations}. In Theorem~\ref{thm:interp-qff}, the logarithmic factor comes from repeating the routine over the randomly drawn parameters, which we perform classically to keep the depth minimal. This outer loop could instead be replaced by amplitude amplification, reducing the overhead at the cost of extra depth, or alternatively by its interpolated version from Theorem~\ref{thm:interp-aa}, giving algorithms with two independent quantum-classical interpolation parameters.

\section{Preliminaries}\label{sec:preliminaries}

\subsection{Random walks and Markov chains}\label{sec:randomwalks}\label{ssec:markov}

A random walk is one of the simplest ways of exploring a graph when no global
information is available \cite{levin2017markov,norris1998markov,lovasz1993random}. Their behaviour is governed entirely by the transition probabilities between
vertices, with each transition depending only on the current vertex. As such, random walks constitute a special case of the framework of Markov chains, which is well studied and understood.

For a finite state space \(X\), with \(n:=|X|\), a time-homogeneous
Markov chain can be represented by an \(n\times n\) row-stochastic matrix
\(P\), where \(P_{xy}\in[0,1]\) is the probability of moving from \(x\) to
\(y\) in one step, and
\[
\sum_{y\in X}P_{xy}=1
\qquad\text{for every }x\in X.
\]
We represent probability distributions on \(X\) as row vectors, so that
one step of the chain maps \(p\) to \(pP\). The matrix \(P\) also defines a
directed weighted graph on \(X\), with an edge from \(x\) to \(y\) whenever
\(P_{xy}>0\), weighted by the corresponding transition probability, which we denote by $G=(X,E)$.

A Markov chain is called \emph{irreducible} if every state can be reached from every other in finitely many steps (i.e., for every $x,y\in X$, there exists $t\ge0$
such that $(P^t)_{xy}>0$), and \emph{aperiodic} if no integer $k>1$ divides the length of every directed cycle of the underlying graph. A chain that is both is called \emph{ergodic}. By the Perron--Frobenius Theorem \cite{seneta2006nonnegative}, every eigenvalue of $P$ has absolute value at most one, and $1$ is an eigenvalue. If $P$ is irreducible, the normalised left $1$-eigenvector is unique and strictly positive, so it may be written as a probability distribution $\pi$ satisfying $\pi P=\pi$. If $P$ is in addition aperiodic, every other eigenvalue is strictly smaller than one in absolute value. The distribution $\pi$ is called the \emph{stationary distribution}. For an ergodic chain it is unique, and $pP^{t}\to\pi$ for every starting distribution $p$ \cite{levin2017markov}.

We now turn to the search problem. Fix a non-empty set $M\subseteq X$ of $m:=|M|$ marked vertices and write $U:=X\setminus M$ for the unmarked ones. The probability that a vertex drawn from the stationary distribution is marked is
\begin{equation}\label{eq:pM}
p_M:=\sum_{x\in M}\pi_x\,.
\end{equation}

Arguably, the simplest way to find a marked vertex is to run the random walk until it lands on one, motivating the following algorithm.

\begin{algorithm}[H]
\caption{Classical random walk search}\label{alg:random-walk}
\begin{enumerate}\itemsep2pt
\item Draw $x\in X$ from the stationary distribution $\pi$.
\item If $x$ is marked, output $x$ and stop.
\item Otherwise take one step of $P$ from $x$ and return to step 2.\label{alg:rw-step}
\end{enumerate}
\end{algorithm}

\begin{definition}
\label{def:HT}
The \emph{hitting time} $\HT(P,M)$ is the expected number of executions of step 3 of Algorithm~\ref{alg:random-walk}, conditioned on the initial vertex being unmarked.
\end{definition}
If the initial stationary distribution is sufficiently dense in marked vertices, the initial draw already yields a marked vertex with constant probability. As such, we henceforth assume that $p_M\le1/2$.

\subsection{Interpolated walks}\label{sec:interpolated}\label{sec:chains}\label{ssec:semiabsorbing}

There is more than one way to look at this same process. Rather than following a single walker who stops when it lands on a marked vertex, we can follow the whole distribution under a modified chain in which the marked vertices do not allow the walker to leave.

\begin{definition}
\label{def:absorbing}
The \emph{absorbing walk} $P'$ is obtained from $P$ by replacing every outgoing transition from a marked vertex by a self-loop,
\[
P'_{xy}=
\begin{cases}
P_{xy}  & \text{if } x\notin M,\\
\delta_{xy} & \text{if } x\in M.
\end{cases}
\]
Ordering the states so that vertices in $U$ come first, we get the block structure
\begin{equation}\label{eq:block-P-Pprime}
P=\begin{pmatrix} P_{UU} & P_{UM}\\ P_{MU} & P_{MM}\end{pmatrix},
\qquad
P'=\begin{pmatrix} P_{UU} & P_{UM}\\ 0 & I\end{pmatrix}.
\end{equation}
\end{definition}

While a classical search can simply run the absorbing chain $P'$ until absorption, a quantum walk is unitary and cannot irreversibly funnel amplitude into $M$. Instead, many modern quantum walk search algorithms consider a one-parameter family of chains that interpolates between $P$ and $P'$ \cite{kmor}.

\begin{definition}
\label{def:interpolated}
For $s\in[0,1]$ the \emph{interpolated Markov chain} is
\begin{equation}\label{eq:Ps}
P(s):=(1-s)\,P+s\,P',
\end{equation}
which in the block structure of \eqref{eq:block-P-Pprime} reads as
\begin{equation}\label{eq:Ps-block}
P(s)=
\begin{pmatrix}
P_{UU} & P_{UM}\\
(1-s)\,P_{MU} & (1-s)\,P_{MM}+s\,I
\end{pmatrix}.
\end{equation}
\end{definition}

In particular, the unmarked block remains unchanged. What the interpolation parameter $s$ achieves is to weaken the outgoing transitions from marked vertices and redirect the missing weight into self-loops, so that $s$ measures how hard it is to leave $M$ once we are there. At $s=0$ the marked vertices are just like any other and the walk is $P$. At $s=1$ they perfectly trap any walker and the walk is $P'$. In between, they possess a certain degree of ``stickiness'', motivating the alternative designation of $P(s)$ as a semi-absorbing walk.

We use both terms interchangeably, although \emph{interpolation} may unrelatedly refer either to the semi-absorbing parameter $s$ or to the quantum-classical parameter $\gamma$, whichever is evident from the context and otherwise stated.

Whilst classical random walks can traverse any directed graph, quantising a Markov chain typically requires an underlying symmetry to ensure that the transition operator admits a purely real spectrum. This is guaranteed by reversibility.

An ergodic Markov chain $P$ with stationary distribution $\pi$ is \emph{reversible} if it satisfies the detailed balance equation,
\begin{equation}\label{eq:detailed-balance}
\pi_x P_{xy}=\pi_y P_{yx}\qquad\text{for all }x,y\in X.
\end{equation}

Reversibility allows one to study the chain through a symmetric operator which shares its eigenvalues and singular values, up to sign. The \emph{discriminant matrix} of a Markov chain $P$ is defined as
\begin{equation}\label{eq:disc}
D(P):=\sqrt{P\circ P^{\mathsf T}}\,,
\end{equation}
where the Hadamard product $\circ$ and the square root are taken entrywise, so that $D_{xy}=\sqrt{P_{xy}P_{yx}}$.
By detailed balance it follows that
\begin{equation}
D(P) = \operatorname{diag}(\sqrt{\pi}) P \operatorname{diag}(\sqrt{\pi})^{-1},
\end{equation}
where $\sqrt{\pi}$ denotes the entrywise square root of $\pi$. Therefore, $D(P)$ is similar to $P$ and has entirely real eigenvalues \cite{levin2017markov}.

Crucially, if $P$ is ergodic and reversible, then so is the semi-absorbing walk $P(s)$ for every $s \in [0,1)$ \cite[Props.~7 and 12]{kmor}.  Consequently, all properties established for reversible, ergodic Markov chains carry over to $P(s)$ for any $s \in [0, 1)$ and, for notational convenience, we will often suppress the explicit dependence on $s$.

Because $D(s):=D(P(s))$ is real and symmetric, its spectrum is real and admits an orthonormal eigendecomposition
\begin{equation}\label{eq:Ds-spectral}
D(s)=\sum_{k=1}^{n}\lambda_k(s)\,\ket{v_k(s)}\!\bra{v_k(s)},
\qquad
\lambda_1(s)\le\cdots\le\lambda_n(s).
\end{equation}
It is convenient to introduce the unit vectors
\begin{equation}\label{eq:U-M-vectors}
\ket{U}:=\frac{1}{\sqrt{1-p_M}}\sum_{x\in U}\sqrt{\pi_x}\,\ket{x},
\qquad
\ket{M}:=\frac{1}{\sqrt{p_M}}\sum_{x\in M}\sqrt{\pi_x}\,\ket{x}.
\end{equation}
These are orthogonal and supported on the unmarked and marked parts of $\mathcal H:=\operatorname{span}\{\ket x:x\in X\}$, respectively.

\begin{proposition}[Propositions~4, 15 and~16 of \cite{kmor}]
\label{prop:principal}\label{prop:D-spectrum}
The largest eigenvalue of $D(s)$ is $1$. It is simple for every $s\in[0,1)$ and has multiplicity $m=|M|$ at $s=1$.
For $s\in[0,1]$ the top (or principal) eigenvector of $D(s)$ may be chosen as $\ket{v_n(s)}=\sqrt{\pi(s)^{\mathsf T}}$, and it decomposes as
\begin{equation}\label{eq:vn-decomp}
\ket{v_n(s)}=\cos\theta(s)\,\ket{U}+\sin\theta(s)\,\ket{M},
\end{equation}
where
\begin{equation}\label{eq:cos-sin-theta}
\cos\theta(s)=\sqrt{\frac{(1-s)(1-p_M)}{1-s(1-p_M)}},
\qquad
\sin\theta(s)=\sqrt{\frac{p_M}{1-s(1-p_M)}}.
\end{equation}
\end{proposition}
Thus, $\ket{v_n(s)}$ lives in the plane spanned by $\ket U$ and $\ket M$, and as $s$ runs from $0$ to $1$ it rotates within that plane from
\begin{equation}\label{eq:vn-rotation}
\ket{v_n(0)}=\sqrt{1-p_M}\,\ket U+\sqrt{p_M}\,\ket M
\qquad\text{to}\qquad
\ket{v_n(1)}=\ket M ,
\end{equation}
\begin{figure}[htbp]
  \centering
  \includegraphics[width=0.40\linewidth]{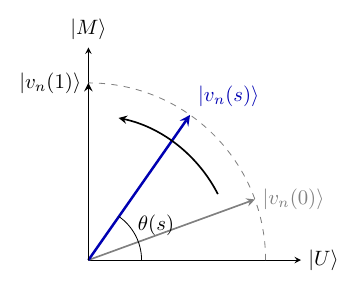}
  \caption{Schematic representation of the principal eigenvector
  $\ket{v_n(s)}=\cos\theta(s)\ket{U}+\sin\theta(s)\ket{M}$.
  As $s$ increases from $0$ to $1$, the angle increases from
  $\theta(0)=\arcsin\sqrt{p_M}$ to $\pi/2$, and the eigenvector
  approaches $\ket{M}$.}
  \label{fig:interpolated-eigenvector}
\end{figure}
as represented in Figure~\ref{fig:interpolated-eigenvector}.\\
When $p_M$ is small, the starting point is almost $\ket U$ and the endpoint is exactly $\ket{M}$. The interpolation parameter then acts as a dial that rotates one into the other, and for intermediate $s$, the vector has substantial overlap with both.

For an ergodic, reversible Markov chain $P$, the classical hitting time to $M$ admits a spectral formulation in terms of the absorbing walk $P'$.

\begin{proposition}[Proposition~2 of \cite{kmor}]
\label{prop:HT-spectral}
Let $D':=D(P')$ denote the discriminant matrix of the absorbing walk, with sorted eigenvalues $\lambda'_1 \le \cdots \le \lambda'_n = 1$ and orthonormal eigenvectors $\ket{v'_k}$. The classical hitting time satisfies
\begin{equation}\label{eq:HT-spectral}
\HT(P,M)=\sum_{k=1}^{n-m}\frac{\bigl|\braket{v'_k|U}\bigr|^{2}}{1-\lambda'_k}\,.
\end{equation}
\end{proposition}

By analogy, this notion can be extended to the non-principal modes of $D(s)$.

\begin{definition}
\label{def:ext-HT}
For $s\in[0,1)$, the \emph{interpolated hitting time} is
\begin{equation}\label{eq:HTs}
\HT(s):=\sum_{k=1}^{n-1}\frac{\bigl|\braket{v_k(s)|U}\bigr|^{2}}{1-\lambda_k(s)}\,,
\end{equation}
and the \emph{extended hitting time} of $P$ with respect to $M$ is
\begin{equation}\label{eq:HT+}
\HTp(P,M):=\lim_{s\to1}\HT(s)\,.
\end{equation}
\end{definition}

While expression \eqref{eq:HT-spectral} stops at $k=n-m$, the sum in \eqref{eq:HTs} runs up to $k=n-1$. For the intermediate modes $k \in \{n-m+1, \dots, n-1\}$, both the numerator and denominator vanish as $s \to 1^-$ due to the emerging $m$-fold degeneracy of the eigenvalue $1$, whose limits need not vanish. Consequently, $\HTp(P,M)$ can strictly exceed $\HT(P,M)$ whenever $|M| > 1$, but is equal to $\HT(P,M)$ for a single marked vertex \cite{kmor}. Furthermore, both hitting times are related in closed form as follows.

\begin{theorem}[Theorem~4 of \cite{kmor}]
\label{thm:HT-s-relation}
For any $s \in [0,1)$,
\begin{equation}\label{eq:HT-s-relation}
\HT(s) = \frac{p_M^2}{(1 - s(1 - p_M))^2}\,\HTp(P,M) = \sin^4\theta(s)\,\HTp(P,M).
\end{equation}
\end{theorem}
The prefactor increases monotonically in $s$ and equals one at $s=1$, so $\HT(s)\le\HTp(P,M)$ throughout. Therefore, raising $s$ improves the overlap of $\ket{v_n(s)}$ with the marked subspace by \eqref{eq:cos-sin-theta}, but it also raises the hitting time, by \eqref{eq:HT-s-relation}.

For the quantum constructions below, we will assume a non-negative spectrum. This can be ensured by replacing $P$
with the \emph{lazy} walk $(P+I)/2$, which stays at the current vertex
with probability $1/2$ and otherwise takes a step of $P$.
The chain $(P+I)/2$ has all its eigenvalues in $[0,1]$, and its hitting time with respect to any $M$ is exactly $2\HT(P,M)$ \cite{levin2017markov}. Since we are mostly interested in asymptotic behaviour, we may, whenever needed, make the walk lazy without altering complexity.

\subsection{The quantum walk operator}\label{ssec:walkop}

Szegedy's framework \cite{Sze04a} associates a discrete-time unitary walk with any reversible Markov chain by embedding its transition dynamics into the bipartite Hilbert space $\mathcal{H} \otimes \mathcal{H}$. Each row distribution of $P$ is coherently encoded onto an auxiliary register, yielding a unitary operator whose eigenphases are governed directly by the spectrum of the classical discriminant matrix.

Formally, let $\ket{\bar 0}\in\mathcal H$ denote a fixed reference state, and define the reference subspace and its associated reflection by
\begin{equation}\label{eq:ref-subspace}
\mathcal{X} := \mathcal{H} \otimes \operatorname{span}\{\ket{\bar 0}\}, \qquad
\Pi_0 := I \otimes \ket{\bar 0}\bra{\bar 0}, \qquad
\operatorname{ref}_{\mathcal X} := 2\Pi_0 - I \otimes I.
\end{equation}
so that $\refX$ reflects about the subspace in which the second register carries the reference state.

The chain enters through a single unitary, $V(P)$, acting on $\mathcal H\otimes\mathcal H$ and satisfying
\begin{equation}\label{eq:V}
V(P)\ket{x}\ket{\bar 0}=\ket{x}\ket{p_x},
\qquad
\ket{p_x}:=\sum_{y\in X}\sqrt{P_{xy}}\,\ket{y},
\qquad\forall x\in X.
\end{equation}
Locality is enforced by the $\Shift$ operation, which exchanges the two registers along edges of the graph and acts as the identity otherwise
\begin{equation}\label{eq:shift}
\Shift\ket{x,y}=
\begin{cases}
\ket{y,x}, & \text{if } x=y \text{ or } \{x,y\}\in E,\\
\ket{x,y}, & \text{otherwise}.
\end{cases}
\end{equation}
Since $P_{xy}=0$ whenever $xy$ is not an edge, $\Shift$ always succeeds on states of the form $\ket{x}\ket{p_x}$.
\begin{definition}
\label{def:W}
The \emph{quantum walk operator} associated with $P$ is
\begin{equation}\label{eq:W}
W(P) := U_D \refX = V(P)^\dagger \Shift V(P) \refX,
\end{equation}
where $U_D := V(P)^\dagger \Shift V(P)$.
\end{definition}
Alternatively, one may consider a more general definition where $U_D$ is replaced
by any unitary $U$ such that $D(P)$ is a diagonal block of
$(I\otimes\bra{\bar 0})\,U\,(I\otimes\ket{\bar 0})$, and may thus act on a
larger space as in Ref.~\cite{apers2021unified}. For our purposes, the following
observation suffices.
\begin{lemma}[\cite{Sze04a}, Proposition~23 of \cite{kmor}]
\label{lem:blockencoding}
For a reversible Markov chain $P$, the operator $U_D$ is Hermitian and self-inverse. Moreover,
\begin{equation}\label{eq:blockencoding}
\Pi_0 U_D \Pi_0 = D(P) \otimes \ket{\bar 0}\!\bra{\bar 0},
\end{equation}
rendering $U_D$ a $(1, \log |X|)$-block-encoding of $D(P)$.
\end{lemma}

Because $P(s)$ remains reversible and ergodic for every $s \in [0,1)$, the construction applies identically to it. In particular, $W(s):=W(P(s))$ has the $(+1)$-eigenvector
\begin{equation}\label{eq:Psi-n-s}
\ket{\Psi_n(s)}=\ket{v_n(s),\bar 0}
=\bigl(\cos\theta(s)\,\ket{U}+\sin\theta(s)\,\ket{M}\bigr)\ket{\bar 0},
\end{equation}
which is the quantum counterpart of the stationary distribution.
\subsection{A spatial search cost model}\label{sec:costmodel}

Inspired by Algorithm~\ref{alg:random-walk}, we may define the following three black-box primitives:
\begin{itemize}
    \item $\Setup(P)$ with cost $\mathtt{S}$: draw a sample from the stationary distribution $\pi$ of $P$;
    \item $\Update(P)$ with cost $\mathtt{U}$: perform one step of $P$;
    \item $\Checkk(M)$ with cost $\mathtt{C}$: check if a given vertex is marked.
\end{itemize}
By this convention, and by definition of the hitting time, Algorithm~\ref{alg:random-walk} finds a marked vertex with constant probability in complexity of order
\begin{equation}
\mathtt{S}+\HT(P,M)(\mathtt{U}+\mathtt{C}).
\end{equation}
When considering the Szegedy walks previously introduced, which are built in analogy with reversible Markov chains, a natural cost model arises from extending the classical primitives to their quantum counterparts:
\begin{itemize}\itemsep2pt
\item $\Setup(P)$, with cost $\mathsf{S}$: prepares the state $\ket{\pi}=\sum_{x\in X}\sqrt{\pi_x}\,\ket{x}$;
\item $\Update(P)$, with cost $\mathsf{U}$: applies any of $V(P)$, $V(P)^{\dagger}$ and $\Shift$, in relation to a fixed reference state $\ket{\bar 0}\in \mathcal{H}$;
\item $\Checkk(M)$, with cost $\mathsf{C}$: applies the marked-vertex oracle
\begin{equation}\label{eq:check}
\Checkk(M)\,\ket{x}\ket{b}=
\begin{cases}
\ket{x}\ket{b\oplus1}, & x\in M,\\[2pt]
\ket{x}\ket{b}, & x\notin M.
\end{cases}
\end{equation}
\end{itemize}
The projection $\Pi_0$ is assumed to be free in this cost model.
Still, the three operations behave quite differently in practice. The quantum $\Setup$, in particular, strays the most from the classical analogue. In fact, preparing the coherent superposition $\ket{\pi}$ is a strictly stronger requirement than merely drawing a sample from $\pi$. Some distributions can be easy to sample while hard to prepare coherently, making $\mathsf{S}$ possibly the limiting quantity.

In order to build $W(s)$, we require access to the interpolated transition operator $V(s)$. Provided the diagonal entries $P_{xx}$ are known on $M$, $V(s)$ decomposes efficiently into the native primitives.

\begin{proposition}[cf. Lemma 3 of \cite{kmor}]
\label{prop:walkstep-cost}
One application of $U_D(s)$ or $W(s)$ costs $O(\mathsf{U}+\mathsf{C})$, uniformly in $s\in[0,1)$.
\end{proposition}

We consider algorithms where a coherent run prepares an initial state ($\ket{\pi}$ via $\Setup$), applies a sequence of unitaries, and ends by measuring all registers, so that only its classical outcome decides whether and how to start the next run. The \emph{quantum-walk depth} of such an algorithm is the largest number of (controlled) walk operator applications ($W(s)$, $W^\dagger(s)$, $U_D(s)$ or $U_D^\dagger(s)$) within any of its coherent runs. Its \emph{hybrid complexity} counts the total cost of all calls to $\Setup$, $\Update$ and $\Checkk$ across all runs, which by Proposition~\ref{prop:walkstep-cost} is at most the number of runs times $\mathsf S$ plus the total number of walk steps times $O(\mathsf U+\mathsf C)$.

This is, in a sense, a generalisation of the hybrid query models used in~\cite{sun2019hybrid,hamoudi2024nisq,ambainis2026tight}, where coherent runs of at most $q$ oracle calls are separated by full measurements and the cost accounts for the total number of calls. Because the $\mathsf{Check}$ operation acts as the marked-set oracle and a depth $D$ run uses $O(D)$ checks, query lower bounds with $q = O(D)$ transfer directly to our checking cost. Unlike the query model, however, our hybrid complexity also charges for $\Setup$ and $\Update$, and hence depends on the number of runs and not only on the number of checks.

\subsection{Quantum fast-forwarding}
\label{sec:qff}

A central primitive for modern quantum walk search algorithms is Quantum Fast-Forwarding (QFF) \cite{ApersSarlette2019}. It allows for a quadratic ``acceleration'' of the simulation of reversible Markov dynamics by coherently realising $t$ steps of the chain with only $\widetilde{O}(\sqrt{t})$ quantum walk steps. Given that the walk operator provides a block-encoding of the discriminant matrix $D$, this is achieved by using QSVT to block-encode the monomial $D^t$ with a Chebyshev polynomial concentrated on a band of size $O(\sqrt{t})$.

We state the result with an operator-norm error guarantee $\varepsilon'$, rather than the original description in terms of normalised state error $\varepsilon$ for the target state $\ket{D^t v}:=D^t\ket{v}/\|D^t\ket{v}\|$, which can be obtained by setting $\varepsilon \ge 2\varepsilon'/\lVert D^t\ket{v}\rVert$.

\begin{theorem}[cf. Theorem 2 of \cite{ApersSarlette2019}]
\label{thm:qff}
Let $P$ be a reversible Markov chain on a finite state space $X$ with discriminant $D$. For any integer target time $t\in\mathbb{N}$ and precision $\varepsilon'\in(0,1)$, there is a quantum subroutine $\hyperlink{alg:qff}{\QFF}(\ket{v}, P, t; \varepsilon')$ that makes $O(\tau)$ applications of $W(P)$ and its inverse, with
\begin{equation}\label{eq:qff-operator}
\tau=\Bigl\lceil\sqrt{2t\ln(2/\varepsilon')}\Bigr\rceil,
\end{equation}
and block-encodes an operator $g_{t,\tau}(D)$ satisfying
\begin{equation}\label{eq:qff-bound}
\bigl\lVert g_{t,\tau}(D)-D^t\bigr\rVert\le\varepsilon'.
\end{equation}
On a normalised input $\ket v\in\operatorname{span}\{\ket{x} : x \in X\}$, \QFF{} succeeds with probability $\lVert g_{t,\tau}(D)\ket v\rVert^2$ and postselecting the successful ancilla block outputs a normalised state proportional to $g_{t,\tau}(D)\ket v$.
\end{theorem}

\section{Proof of Theorem~\ref{thm:cq_interp}: single marked vertex}
\label{sec:proof_kmor}
\subsection{KMOR search in quantum fast-forwarding form}\label{sec:algorithm}

We now detail the approach of KMOR, although from a slightly different
perspective. From Proposition~\ref{prop:principal}, the principal
eigenvector of $D(s)$ decomposes as
$\ket{v_n(s)}=\cos\theta(s)\ket U+\sin\theta(s)\ket M$.
Choosing $s$ so that this vector has substantial overlap with both
$\ket U$ and $\ket M$ suggests a simple search procedure: we start from
$\ket U$, approximately project onto $\ket{v_n(s)}$, and measure for a
marked vertex. An exact projection would then yield a marked vertex with
probability $\cos^2\theta(s)\sin^2\theta(s)$, which is maximised when
$\cos\theta(s)=\sin\theta(s)=1/\sqrt2$.

The cost therefore depends on how this approximate projection is implemented.
Krovi, Magniez, Ozols and Roland~\cite{kmor} use coherent eigenvalue estimation
on $W(s)$ to isolate its principal eigenstate. Alternatively, one might think to construct a minimax polynomial approximating the indicator of an eigenvalue
as done in Ref.\cite{lintong2020}. This approach, however, requires a lower bound on its spectral gap from the remaining spectrum, which is not assumed here. Moreover, suppressing all
non-principal modes uniformly does not account for their overlap
with $\ket U$, which is necessary to recover the hitting-time dependence.
Here, we instead argue that
QFF can also be used as an eigenfilter. Advancing the interpolated chain for $t$ steps
applies the spectral transformation $\lambda\mapsto\lambda^t$, which
preserves the principal eigenvector and suppresses the remaining modes
\begin{equation}\label{eq:kmor-filter}
D(s)^t\ket U
=
\cos\theta(s)\ket{v_n(s)}
+
\sum_{k=1}^{n-1}
\lambda_k(s)^t\braket{v_k(s)|U}\ket{v_k(s)}.
\end{equation}
The residual contribution depends both on how close the eigenvalues
are to one and on their overlap with the initial state. These are
precisely the quantities which dictate $\HT(s)$, allowing the filtering
time to be controlled by the hitting time rather than by the spectral
gap alone. The speedup then follows directly from the decaying being fast
forwarded quadratically.

This approach also completely circumvents the use of additional phase registers and of the
Quantum Fourier Transform subroutine, on which textbook formulations of eigenvalue estimation are dependent. This is particularly relevant in the context of quantum-classical
interpolation, where reducing resource requirements is a central concern. Furthermore, it also unifies the most prominent search algorithms as built from the same primitive, differing only in whether the semi-absorbing interpolation parameter is held fixed (KMOR, $\HTp$) or drawn at random (AGJK, $\HT$).

The resulting procedure is obtained by, in step~4(a),  replacing eigenvalue estimation on $W(s)$ by fast-forwarding of the interpolated chain $P(s)$. Note that, in the original KMOR algorithm, $t=\lceil\log T\rceil$ denotes the precision of eigenvalue estimation, so the depth $T$ is restricted to powers of two. Here, $t$ is instead the number of steps of the classical chain to be fast-forwarded, so that the depth $T$ in~\eqref{eq:T-qff} follows directly from QFF. Expressing the algorithm through the classical time $t$ is what allows its depth to be tuned in its interpolation, at the price of an additional (constant) error $\varepsilon_3$ from the polynomial approximation.
\begin{algorithm}[H]
\caption{\textsc{KMORQFF}$(P, M, s, t, \varepsilon_3)$}\label{alg:kmorqff}

\noindent
\begin{enumerate}
	\item Use $\Setup(P)$ to prepare the state
	$\ket{\pi}\ket{\bar{0}}$ in registers $\mathsf{R}_1\mathsf{R}_2$.
	\item Attach an ancilla $\mathsf{R}_3$ initialised to $\ket{0}$, apply
	$\Checkk(M)$ to $\mathsf{R}_1\mathsf{R}_3$, and measure
	$\mathsf{R}_3$.
	\item If $\mathsf{R}_3 = 1$: measure $\mathsf{R}_1$ in the vertex basis
	and output the result.
	\item Otherwise, discard $\mathsf{R}_3$ and:
	\begin{enumerate}
		\item[(a)] Apply \QFF$\bigl(\ket{U},P(s),t,\varepsilon_3\bigr)$ on
$\mathsf{R}_1\mathsf{R}_2$. If it reports failure, output
``fail''.
		\item[(b)] Attach $\mathsf{R}_3$, apply $\Checkk(M)$ to
		$\mathsf{R}_1\mathsf{R}_3$, and measure $\mathsf{R}_3$.
		\item[(c)] If $\mathsf{R}_3 = 1$: measure $\mathsf{R}_1$ and output
		the result. \\
		Otherwise: output ``no marked vertex.''
	\end{enumerate}
\end{enumerate}
\end{algorithm}

\begin{theorem}[QFF analogue of Theorem 6 of \cite{kmor}]\label{thm:main-qff}
	Assume that the values of $p_M$ and $\HTp(P, M)$ are known, and let
	$s \in [0,1)$, an integer $t \ge 1$, and
	$\tfrac{1}{2} \ge \varepsilon_1 \ge \varepsilon_2 +2\varepsilon_3\sin\theta(s) > 0$,
	$\varepsilon_3 \in (0,\tfrac12]$ be parameters satisfying
	\begin{equation}\label{eq:conditions-qff}
		\cos\theta(s)\,\sin\theta(s) \ge \varepsilon_1
		\qquad\text{and}\qquad
		t \ge \frac{\HT(s)}{2\,\varepsilon_2^{2}},
	\end{equation}
	where $\cos\theta(s)$, $\sin\theta(s)$ are given by
	\eqref{eq:cos-sin-theta} and $\HT(s)$ is the interpolated hitting time
	(Definition~\ref{def:ext-HT}). Then
	\KMOR$(P, M, s, t, \varepsilon_3)$ finds a marked vertex with
	success probability at least
	\begin{equation}\label{eq:success-prob-qff}
		p_M + (1 - p_M)\,
		\bigl(\varepsilon_1 - \varepsilon_2
		- 2\varepsilon_3\sin\theta(s)\bigr)^2
	\end{equation}
	and complexity of order $\mathsf{S} + T \cdot (\mathsf{U} + \mathsf{C})$,
	where
	\begin{equation}\label{eq:T-qff}
		T = \Bigl\lceil \sqrt{2t\ln(2/\varepsilon_3)} \Bigr\rceil .
	\end{equation}
\end{theorem}

\begin{proof}
	Step~1 costs $\mathsf{S}$.
	By Theorem~\ref{thm:qff}, step~4(a) uses $O(T)$ applications of
	$W(s)$ and its inverse, each of which costs $O(\mathsf{U} + \mathsf{C})$
	by Proposition~\ref{prop:walkstep-cost}. The total complexity is therefore
	of order $\mathsf{S} + T \cdot (\mathsf{U} + \mathsf{C})$.

	The initial check in steps~2--3 finds a marked vertex with probability
	$p_M$. Conditioned on this failing, the state entering step~4 is
	$\ket{U}\ket{\bar{0}}$. Let $q$ denote the probability that step~4 outputs
	a marked vertex. Then the overall success probability is
	\begin{equation}\label{eq:overall-qff}
		p_M + (1 - p_M)\,q,
	\end{equation}
	so it suffices to show
	$q \ge (\varepsilon_1 - \varepsilon_2 - 2\varepsilon_3\sin\theta)^2$.\\
	Define
	\begin{equation}\label{eq:alpha-k-qff}
		\alpha_k := \braket{v_k(s) | U}.
	\end{equation}
	Expanding $\ket{U}$ in the eigenbasis of $D(s)$ and suppressing the
	dependence on $s$ for readability,
	\begin{equation}\label{eq:U-expansion-qff}
		\ket{U}\ket{\bar{0}}
		= \alpha_n \ket{v_n}\ket{\bar 0}
		+ \sum_{k=1}^{n-1} \alpha_k \ket{v_k}\ket{\bar 0}.
	\end{equation}

	Denote by $g := g_{t,\tau}$ the underlying polynomial applied by the \QFF${}$ circuit in step~4(a), with
$\tau = \lceil\sqrt{2t\ln(2/\varepsilon_3)}\rceil$, which satisfies
\begin{equation}\label{dor}
	\sup_{x\in[-1,1]}\bigl|g(x) - x^{t}\bigr| \le \varepsilon_3\,,
\end{equation}
so that \QFF${}$ maps each eigenvector as
	\begin{align}
		\ket{v_n}\ket{\bar 0} &\;\longmapsto\;
		g(1)\,\ket{v_n}\ket{\bar 0} + \ket{\perp_n},
		\label{eq:qsvt-n} \\
		\ket{v_k}\ket{\bar 0} &\;\longmapsto\;
		g(\lambda_k)\,\ket{v_k}\ket{\bar 0} + \ket{\perp_k},
		\label{eq:qsvt-k}
	\end{align}
	since $\lambda_n = 1$ by Proposition~\ref{prop:D-spectrum}.
	After applying \QFF${}$, the system is in the state
	\begin{equation}\label{eq:Phi-qff}
		\ket{\Phi}
		= \alpha_n\, g(1) \ket{v_n}\ket{\bar 0}
		+ \sum_{k=1}^{n-1} \alpha_k\, g(\lambda_k) \ket{v_k}\ket{\bar 0}
		+ \ket{\perp},
		\qquad \ket{\perp} := \sum_{k} \alpha_k \ket{\perp_k}.
	\end{equation}
	Let $\Pi_M := \sum_{x \in M} \ket{x}\!\bra{x}$ be the projector onto
	marked vertices in the first register. Then
	\begin{align}
		\sqrt{q}
		&= \bigl\|(\Pi_M \otimes \ket{\bar 0}\!\bra{\bar 0})\ket{\Phi}\bigr\|
		\label{eq:q-def-qff} \\
		&= \Bigl\| \Pi_M \Bigl(
		\alpha_n g(1)\ket{v_n} + \sum_{k=1}^{n-1}\alpha_k g(\lambda_k)\ket{v_k}
		\Bigr)\Bigr\|,
		\label{eq:q-proj-qff}
	\end{align}
	since step~4(a) discards the branch in which $\mathsf{R}_2$ is not found
	in $\ket{\bar 0}$, which removes $\ket{\perp}$ exactly.
	By the reverse triangle inequality,
	\begin{equation}\label{eq:triangle-qff}
		\sqrt{q}
		\ge \bigl\|\alpha_n\, g(1)\, \Pi_M \ket{v_n}\bigr\|
		- \left\|
		\sum_{k=1}^{n-1} \alpha_k\, g(\lambda_k) \ket{v_k}
		\right\|,
	\end{equation}
	where we used $\|\Pi_M\| \le 1$ to drop the projector in the second term
	(which only increases it, keeping the lower bound valid).

	By Proposition~\ref{prop:principal},
	$\ket{v_n} = \cos\theta\,\ket{U} + \sin\theta\,\ket{M}$, hence
	$\alpha_n = \braket{v_n | U} = \cos\theta$ and
	$\|\Pi_M \ket{v_n}\| = \sin\theta$. Moreover, at $x=1$ we have $g(1) \ge 1 - \varepsilon_3$, so the first term of
	\eqref{eq:triangle-qff} is at least
	$(1-\varepsilon_3)\cos\theta\,\sin\theta$.

	Since $\ket{v_1},\ldots,\ket{v_{n-1}}$ are mutually orthogonal, the second
	term of \eqref{eq:triangle-qff} satisfies
	\begin{equation}\label{eq:second-term-qff}
		\left\|
		\sum_{k=1}^{n-1} \alpha_k\, g(\lambda_k) \ket{v_k}
		\right\|
		= \sqrt{\sum_{k=1}^{n-1} |\alpha_k|^2\, g(\lambda_k)^2}
		\;\le\;
		\sqrt{\sum_{k=1}^{n-1} |\alpha_k|^2\, \lambda_k^{2t}}
		\;+\; \varepsilon_3\,\sin\theta ,
	\end{equation}
	where we used the triangle inequality together with $|g(\lambda_k)| \le \lambda_k^{t} +
	\varepsilon_3$ from \eqref{dor} and that the residual weight is given by
	\begin{equation}\label{eq:residual-weight}
		\sum_{k=1}^{n-1}|\alpha_k|^2
		= \|\ket{U}\|^2 - |\alpha_n|^2
		= 1 - \cos^2\theta
		= \sin^2\theta\,,
	\end{equation}
	from the decomposition of the $\ket{U}$ state.
	Combining \eqref{eq:triangle-qff} and \eqref{eq:second-term-qff} then yields
	\begin{equation}\label{eq:sqrt-q-bound-qff}
		\sqrt{q}
		\ge (1-\varepsilon_3)\cos\theta\,\sin\theta
		- \sqrt{\sum_{k=1}^{n-1} |\alpha_k|^2\, \lambda_k^{2t}}
		- \varepsilon_3 \sin\theta .
	\end{equation}
	By making the walk lazy, we may assume every $\lambda_k(s)$ lies in $[0,1]$, so that by the elementary inequalities $\lambda \le e^{-(1-\lambda)}$
	and $e^{-y} \le 1/y$ for $y > 0$, we get
	\begin{equation}\label{eq:tail-qff}
		\lambda_k^{2t}
		\le e^{-2t(1-\lambda_k)}
		\le \frac{1}{2t\,(1 - \lambda_k)} .
	\end{equation}
Now recall from Definition~\ref{def:ext-HT} that
	\begin{equation}\label{eq:HT-recall-qff}
		\HT(s) = \sum_{k=1}^{n-1}
		\frac{|\alpha_k|^2}{1 - \lambda_k(s)}.
	\end{equation}
	Then combining \eqref{eq:tail-qff} and \eqref{eq:HT-recall-qff}
	\begin{equation}\label{eq:sum-bound-qff}
		\sum_{k=1}^{n-1} |\alpha_k|^2\, \lambda_k^{2t}
		\le \frac{1}{2t} \sum_{k=1}^{n-1}
		\frac{|\alpha_k|^2}{1 - \lambda_k(s)}
		= \frac{\HT(s)}{2t}.
	\end{equation}
	Finally, substituting \eqref{eq:sum-bound-qff} into \eqref{eq:sqrt-q-bound-qff},
	and using $(1-\varepsilon_3)\cos\theta\sin\theta \ge
	\cos\theta\sin\theta - \varepsilon_3\sin\theta$
	\begin{equation}\label{eq:final-sqrt-q-qff}
		\sqrt{q}
		\ge \cos\theta(s)\,\sin\theta(s)
		- \sqrt{\frac{\HT(s)}{2t}}
		- 2\varepsilon_3\sin\theta(s)
		\ge \varepsilon_1 - \varepsilon_2 - 2\varepsilon_3\sin\theta(s),
	\end{equation}
	where the last step uses the two hypotheses in
	\eqref{eq:conditions-qff}. Squaring gives
	$q \ge (\varepsilon_1 - \varepsilon_2 - 2\varepsilon_3\sin\theta)^2$, and by
	\eqref{eq:overall-qff} the overall success probability is at least
	\[
	p_M + (1 - p_M)(\varepsilon_1 - \varepsilon_2
	- 2\varepsilon_3\sin\theta(s))^2. \qedhere
	\]
\end{proof}

The maximal value of $\cos\theta(s)\sin\theta(s)$ is $1/2$, attained when
$\theta(s) = \pi/4$, i.e., $\cos\theta(s) = \sin\theta(s) = 1/\sqrt{2}$.
By \eqref{eq:cos-sin-theta} this corresponds to the optimal interpolation
parameter
\begin{equation}\label{eq:s-optimal}
	s^{*}(p_M) = 1 - \frac{p_M}{1 - p_M},
\end{equation}
possible for any $0<p_M\le1/2$, which we previously assumed.\\
At this value, $\cos\theta(s^*)\sin\theta(s^*) = 1/2$, and by
Theorem~\ref{thm:HT-s-relation}, $\HT(s^*) = \tfrac{1}{4}\HTp(P, M)$. Setting
$\varepsilon_1 = 1/3$, $\varepsilon_2 = 1/12$ and $\varepsilon_3 = \sqrt{2}/24$, the
conditions in \eqref{eq:conditions-qff} are satisfied whenever
$t \ge 18\,\HTp(P, M)$, and since $\sin\theta(s^*) = 1/\sqrt2$, the success
probability is at least
\[
p_M + (1 - p_M) \cdot
\Bigl(\frac{1}{3} - \frac{1}{12} - \frac{1}{12}\Bigr)^{\!2}
\;=\; p_M + \frac{1-p_M}{36}
\;\ge\; \frac{1}{36}.
\]

\begin{corollary}\label{cor:main-qff}
	Assume that the values of $p_M\le1/2$ and $\HTp(P, M)$ are known, and let
	$s^* = 1 - p_M/(1-p_M)$.
	If $t=\lceil18\,\HTp(P,M)\rceil$, then
	\KMOR$(P, M, s^*, t, \sqrt{2}/24)$ finds a marked vertex with
	probability at least $1/36$, at depth
	$T = \lceil\sqrt{2t\ln(24\sqrt 2)}\,\rceil = O(\sqrt{\HTp(P,M)})$ and complexity of
	order $\mathsf{S} + T \cdot (\mathsf{U} + \mathsf{C})$.
	In particular, when $|M| = 1$, $\HTp(P, M) = \HT(P, M)$
	and the algorithm achieves a full
	quadratic speedup over the classical hitting time.
\end{corollary}

\subsection{The interpolation theorem}
We now show that the KMOR search algorithm admits a smooth quantum-classical
interpolation by suitably choosing the semi-absorbing walk interpolation parameter $s$ and the
fast-forwarding time $t$ jointly as functions of a quantum-classical interpolation parameter $\gamma \in [0,1]$. The same arguments can be applied with minimal changes to the original KMOR eigenvalue estimation formulation. A lower $\gamma$ restricts the maximum allowed depth $D(\gamma)$, thereby reducing the number of quantum walk steps and consequently giving a smaller per-instance success probability, which is then compensated
by classical repetition.

Both the success probability and the required depth are controlled by the angle
$\theta(s)$, which dictates the decomposition of the principal eigenvector
$\ket{v_n(s)}$ in $\mathrm{span}\{\ket{U},\ket{M}\}$. On the one hand, it sets the
success amplitude $\cos\theta(s)\sin\theta(s)$, which is of order $\sin\theta(s)$
for $s \le s^*$. On the other hand, it controls the hitting time through
$\mathrm{HT}(s) = \sin^4\theta(s)\cdot\mathrm{HT}^+$. Decreasing $s$ therefore
makes the filtering cheaper, as the number of steps $t$ that must be
fast-forwarded grows with $\mathrm{HT}(s)$. However, since the success amplitude
also decreases, the residual contribution of the non-principal modes must be
suppressed to a correspondingly lower level, which partially offsets this gain.
Overall, the required depth decreases linearly in $\sin\theta(s)$, whereas the
success probability decreases quadratically, so that a reduction in quantum-walk
depth can be traded for a larger number of classical repetitions. We therefore choose $s$ so that the principal eigenvector remains closer to
$\ket{U}$ than at $s^*$, reducing the depth at the cost of a lower success
probability per run.
All of these trade-offs are illustrated by (quantum-classical) interpolating the (semi-absorbing) interpolation parameter $s_\gamma$ as shown in Figure~\ref{fig:interp-arc}.

\begin{figure}[t]
  \centering
  \includegraphics[width=0.67\linewidth]{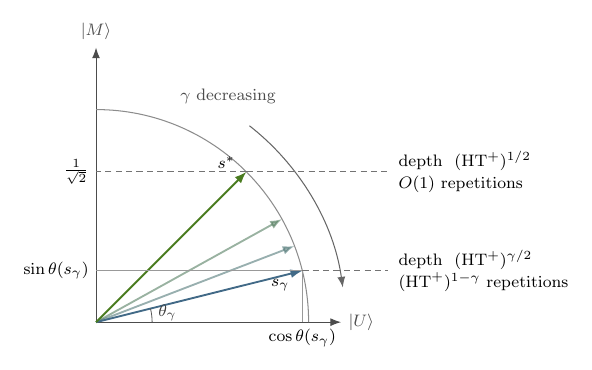}
  \caption{Geometric illustration of the quantum-classical interpolation of the semi-absorbing parameter $s_\gamma$. The principal eigenvector
  $\ket{v_n(s)}=\cos\theta(s)\ket{U}+\sin\theta(s)\ket{M}$ moves along the arc as $s$ decreases and the angle $\theta$ controls both the marked amplitude and the interpolated hitting time
  $\mathrm{HT}(s)=\sin^4\!\theta\,\mathrm{HT}^+$. Moving toward $\ket{U}$ therefore reduces the coherent depth
  (linearly in $\sin\theta$) at the cost of quadratically increasing the number of repetitions.}
  \label{fig:interp-arc}
\end{figure}

\begin{namedresult}{Theorem~\ref{thm:cq_interp} (general form)}
	Let $P$ be an ergodic, reversible Markov chain on $G = (X,E)$, let
	$M \subseteq X$ be a non-empty set of marked vertices with
	$p_M \leq \tfrac{1}{2}$, and suppose the values of $p_M$ and
	$\HTp:=\HTp(P,M)$ are known.  For any $\gamma \in [0,1]$, there exists
	a quantum walk-based algorithm with quantum-walk depth
	\begin{equation}\label{eq:Tgamma}
		D(\gamma) \;=\; O\left((\HTp)^{\gamma/2}\right)
	\end{equation}
	which finds a marked vertex with constant success probability
	and hybrid complexity of order at most
	\begin{equation}\label{eq:complexity_gamma}
		(\HTp)^{1-\gamma} \cdot \mathsf{S}
		\;+\; (\HTp)^{1-\gamma/2} \cdot (\mathsf{U} + \mathsf{C}).
	\end{equation}
\end{namedresult}

\begin{proof}
	Throughout the proof we abbreviate $H := \HTp(P,M)$.
	Since $\gamma \in [0,1]$, note that the quantity $H^{\gamma-1}$ ranges over
	$[1/H,\, 1]$. The proof is split into three regimes according to its relationship
	with~$p_M$.

	Suppose first that $p_M > H^{\gamma-1}$.  Then the marked vertices are
	dense enough to be found by independent sampling. Each draw from $\pi$ via
	$\Setup(P)$ lands on a marked vertex with probability
	$p_M = \Omega(H^{\gamma-1})$, so $O(1/p_M) = O(H^{1-\gamma})$ independent
	trials, each costing $\mathsf{S} + \mathsf{C}$, succeed with constant
	probability.  The total complexity of finding a marked vertex in this case is then
	$O(H^{1-\gamma}) \cdot (\mathsf{S} + \mathsf{C})$.  Since
	$1 - \gamma \leq 1 - \gamma/2$ for all $\gamma \in [0,1]$, we have
	$H^{1-\gamma} \leq H^{1-\gamma/2}$, so
	$H^{1-\gamma}\cdot\mathsf{C}
	\leq H^{1-\gamma/2}\cdot(\mathsf{U}+\mathsf{C})$, which fits in the claimed
	bound~\eqref{eq:complexity_gamma}.

	Suppose instead that $p_M \leq H^{\gamma-1}$, and consider first the
	sub-regime $H^{\gamma-1} \geq \tfrac{1}{2}$.  Equivalently,
	$H^{1-\gamma} \leq 2$, which forces
	$\gamma \geq 1 - \log 2 / \log H$, meaning we are exponentially close to the
	full-quantum regime and full-depth search fits
within the desired depth bound up to a constant factor. Invoking Corollary~\ref{cor:main-qff} with
	$s = s^*$ from \eqref{eq:s-optimal} and $t = 18\,H$ finds a
	marked vertex with probability at least $1/36$ and complexity
	$O\bigl(\mathsf{S} + \sqrt{H}\,(\mathsf{U}+\mathsf{C})\bigr)$ at depth
	$O(\sqrt{H})$.  Both the depth and the per-call cost are absorbed
	into the claimed bounds.

	It remains to treat the interesting regime where
	$p_M \leq H^{\gamma-1} < \tfrac{1}{2}$.  By
	Proposition~\ref{prop:principal}, $\sin^2\!\theta(s)$ is continuous and
	strictly increasing on $[0,1)$, with $\sin^2\!\theta(0) = p_M$ and
	$\sin^2\!\theta(s) \to 1$ as $s \to 1^-$.  Since
	$p_M \leq H^{\gamma-1} < \tfrac{1}{2} = \sin^2\!\theta(s^*)$
	(where $s^*$ is the optimal parameter from \eqref{eq:s-optimal}),
	the intermediate value theorem yields a unique
	$s_\gamma \in [0, s^*)$ satisfying
	\begin{equation}\label{eq:sgamma_def}
		\sin^2\!\theta(s_\gamma) \;=\; H^{\gamma-1}.
	\end{equation}
	Solving explicitly from \eqref{eq:cos-sin-theta}, this amounts to picking
	\begin{equation}\label{eq:sgamma_closed_form}
		s_\gamma \;=\; \frac{H^{\gamma-1} - p_M}{H^{\gamma-1}(1-p_M)}.
	\end{equation}
	At this value, $\sin\theta(s_\gamma) = H^{(\gamma-1)/2}$ and
	$\cos^2\!\theta(s_\gamma) = 1 - H^{\gamma-1} > \tfrac{1}{2}$.
	We now verify that Theorem~\ref{thm:main-qff} applies to
	\KMOR$(P, M, s_\gamma, t_\gamma, \varepsilon_3)$ with parameters
	\begin{equation}\label{eq:choice_eps}
		\varepsilon_1
		\;:=\; \cos\theta(s_\gamma)\,\sin\theta(s_\gamma),
		\qquad
		\varepsilon_2
		\;:=\; \frac{\sqrt{2}}{12}\, H^{(\gamma-1)/2},
		\qquad
		\varepsilon_3 \;:=\; \frac{\sqrt{2}}{24},
		\qquad
		t_\gamma \;:=\; 36\, \ceil{H^{\gamma}}.
	\end{equation}
	The overlap condition
	$\cos\theta(s_\gamma)\sin\theta(s_\gamma) \geq \varepsilon_1$ in
	\eqref{eq:conditions-qff} is immediate from the definition of
	$\varepsilon_1$.  For the fast-forwarding time condition, Theorem~\ref{thm:HT-s-relation}
	and \eqref{eq:sgamma_def} give
	\begin{equation}\label{eq:HTs_gamma}
		\HT(s_\gamma)
		\;=\; \sin^4\!\theta(s_\gamma)\cdot H
		\;=\; H^{2(\gamma-1)}\cdot H
		\;=\; H^{2\gamma-1}.
	\end{equation}
	Substituting into the
	right-hand side of the second condition in \eqref{eq:conditions-qff}
	\begin{equation}\label{eq:depth_check}
		\frac{\HT(s_\gamma)}{2\,\varepsilon_2^{2}}
		\;=\; \frac{H^{2\gamma-1}}{2}\cdot
		\frac{72}{H^{\gamma-1}}
		\;=\; 36\,H^{\gamma}
		\;\le\; t_\gamma,
	\end{equation}
	confirming the second condition.  By \eqref{eq:T-qff} the resulting
	quantum depth is
	\begin{equation}\label{eq:Tgamma_derived}
		T_\gamma
		\;=\; \Bigl\lceil \sqrt{2\,t_\gamma\ln(2/\varepsilon_3)} \Bigr\rceil
		\;=\; \Bigl\lceil
		\sqrt{72\,\ln(24\sqrt2)}\;H^{\gamma/2}
		\Bigr\rceil
		\;=\; O\!\bigl(H^{\gamma/2}\bigr),
	\end{equation}
	matching the depth budget~\eqref{eq:Tgamma}.  Note that $\varepsilon_3$ is
	constant, so no logarithmic factor in $H$ or in $\gamma$ appears from the QFF precision.
	It remains to verify the ordering
	$\tfrac{1}{2} \geq \varepsilon_1
\geq \varepsilon_2 + 2\varepsilon_3\sin\theta(s_\gamma) > 0$ required by
	Theorem~\ref{thm:main-qff}.  The bound $\varepsilon_1 \leq 1/2$ holds since
	$\cos\theta\,\sin\theta = \tfrac{1}{2}\sin(2\theta) \leq 1/2$ for all
	$\theta$.  For $\varepsilon_1 > \varepsilon_2$, recall that
	$\cos\theta(s_\gamma) > 1/\sqrt{2}$, so
	\begin{equation}\label{eq:eps1_lb}
		\varepsilon_1
		\;=\; \cos\theta(s_\gamma)\cdot H^{(\gamma-1)/2}
		\;>\; \frac{H^{(\gamma-1)/2}}{\sqrt{2}}
		\;=\; 6\,\varepsilon_2
		\;>\; \varepsilon_2.
	\end{equation}
	Moreover, since $2\varepsilon_3\sin\theta(s_\gamma)=\varepsilon_2$,
\[
    0 < \varepsilon_2 + 2\varepsilon_3\sin\theta(s_\gamma)
    = 2\varepsilon_2 < \varepsilon_1,
\]
where the last inequality follows from \eqref{eq:eps1_lb}.
	All hypotheses of
	Theorem~\ref{thm:main-qff} are therefore satisfied, and a single call to \KMOR$(P, M, s_\gamma, t_\gamma, \varepsilon_3)$
	finds a marked vertex
	with probability at least
	\begin{align}
		q_\gamma
		&\;\geq\; p_M + (1-p_M)\bigl(\varepsilon_1 - \varepsilon_2
		- 2\varepsilon_3\sin\theta(s_\gamma)\bigr)^2 \notag\\
		&\;\geq\; \bigl(\varepsilon_1 - \varepsilon_2
		- 2\varepsilon_3\sin\theta(s_\gamma)\bigr)^2 \notag\\
		&\;\geq\; \left(
		\frac{6\sqrt2}{12}\,H^{(\gamma-1)/2}
		- \frac{\sqrt2}{12}\,H^{(\gamma-1)/2}
		- \frac{\sqrt2}{12}\,H^{(\gamma-1)/2}
		\right)^{\!2} \notag\\
		&\;=\; \tfrac{2}{9}\,H^{\gamma-1},
		\label{eq:qgamma_bound}
	\end{align}
	where the third line uses $\sin\theta(s_\gamma) = H^{(\gamma-1)/2}$ in
	the last term as well as in $\varepsilon_1$.
	A marked vertex is thus found with constant probability after
	$O(1/q_\gamma) = O(H^{1-\gamma})$ independent repetitions.  Since each
	call has complexity of order
	$\mathsf{S} + T_\gamma\cdot(\mathsf{U}+\mathsf{C})$ by
	Theorem~\ref{thm:main-qff}, the total complexity is
	\begin{equation}\label{eq:case3_total}
		O(H^{1-\gamma})\cdot
		\bigl(\mathsf{S} + H^{\gamma/2}\cdot(\mathsf{U}+\mathsf{C})\bigr)
		\;=\; H^{1-\gamma}\cdot\mathsf{S}
		\;+\; H^{1-\gamma/2}\cdot(\mathsf{U} + \mathsf{C}).
	\end{equation}
\end{proof}
The proof of Theorem~\ref{thm:cq_interp} as stated in Section~\ref{sec:results} then follows by restricting this result to the single marked vertex case, where $\HTp(P,M)=\HT(P,M)$.

\section{Proof of Theorem~\ref{thm:interp-qff}: multiple marked vertices}\label{ssec:interpolated-qff}\label{sec:agjk}
We now generalise the algorithm of Ambainis, Gilyén, Jeffery and Kokainis \cite{ambainisqff} to derive an interpolated version that continuously trades coherent quantum-walk depth against total hybrid complexity. Crucially, it is governed by the ordinary hitting time $\HT(P,M)$ rather than the extended hitting time $\HTp(P,M)$ of the previous section, although it comes with additional logarithmic factors.

This distinction stems from the fact that the KMOR approach solves a harder problem \cite{ambainisqff}. Rather than merely sampling an arbitrary marked vertex, it prepares the coherent stationary state restricted to the marked set $M$. Furthermore, while  both algorithms (as presented here) rely on the QFF subroutine, they diverge fundamentally in their choice of the semi-absorbing parameter. KMOR fixes $s$ to target this asymptotic stationary distribution, whereas AGJK draws $s$ (alongside the evolution time) at random to detect any first-hitting event, without necessarily reaching stationarity.

Most of the original analysis carries over, with the main new ingredient required being a lower bound on the probability that a reversible Markov chain reaches the marked set within a truncated time. Once this is established, the interpolation follows by parameterising the number of walk steps in terms of $\gamma\in[0,1]$.

\subsection{Classical trajectory bounds}

The AGJK analysis studies the marked component of the
fast-forwarded state directly, without requiring convergence to
the principal eigenvector. Its central observation is that this
component can be related to classical trajectories which start
unmarked, visit the marked set, and subsequently return to the
unmarked set, transforming the problem from tracking quantum amplitudes into a combinatorial analysis concerning the classical interpolated walk $P(s)$. The semi-absorbing
parameter controls how long these
trajectories remain marked, allowing a random observation time to
capture their marked portions.

Let
\begin{equation}\label{eq:ff-piU}
    \ket{\pi_U}:=\sum_{x\in U}\sqrt{\pi_x}\ket x
    =\sqrt{1-p_M}\ket U,
\end{equation}
denote the subnormalised restriction of the stationary state $\ket\pi$ to the unmarked subspace.
Starting from this state, the probability of subsequently
passing the QFF ancilla check and measuring a marked vertex is given by
\begin{equation}\label{eq:ff-success-norm}
    \bigl\lVert\Pi_M D(s)^t\ket{\pi_U}\bigr\rVert^2.
\end{equation}
Reversibility allows this squared norm to be interpreted as a
classical trajectory.
\begin{lemma}[Lemma~2 of \cite{ambainisqff}]
\label{lem:ff-bridge}
    Let $P$ be a reversible ergodic Markov chain with stationary
    distribution $\pi$, let $s\in[0,1)$, and let
    $Y(s)=(Y_i(s))_{i\ge0}$ be the Markov chain evolving according to $P(s)$ with
    $Y_0(s)\sim\pi$. Then, for all integers $0\le t\le t'$,
    \begin{equation}\label{eq:ff-bridge}
        \bigl\lVert\Pi_M D(s)^t\ket{\pi_U}\bigr\rVert
        \ge
        \Pr\bigl(Y_0(s)\in U,\;
                 Y_t(s)\in M,\;
                 Y_{t'}(s)\in U\bigr).
    \end{equation}
\end{lemma}

The remaining question is how to choose $s$ so that all three visits are observed with sufficient probability.
Recall that $P(s)$ acts exactly as $P$ on unmarked vertices,
while at a marked vertex it takes a genuine $P$-step only with
probability $1-s$. Equivalently, one may generate $Y(s)$ from $Y$ as follows: run $Y$, and every time it visits a marked vertex, prolong the stay there by an independent holding time that is geometrically distributed with success parameter $1-s$, hence of mean
\begin{equation}\label{eq:ff-r-def}
    r:=\frac{1}{1-s}.
\end{equation}
This provides a useful picture for our choice of $s$. When visits to $M$ are brief, sampling the walk at random times has a low probability of hitting a marked vertex. While increasing $r$ makes marked states easier to reach, it also delays the transitions between unmarked vertices. No single choice of $s$ can resolve this tension across all observation windows,  since the appropriate stretching depends on the trajectory, which is not known in advance.

The solution is to average over geometrically
spaced scales at the cost of logarithmic factors. For some time $T\in\mathbb N$, define
\begin{equation}\label{eq:ff-Rset}
    R:=\bigl\{1,2,4,\ldots,2^{\lceil\log_2(12T)\rceil}\bigr\},
    \qquad
    S:=\bigl\{1-\tfrac1r:r\in R\bigr\}.
\end{equation}
Choosing $s$ from $S$ successfully balances the exploration whenever the underlying chain satisfies the trajectory conditions defined by the following event $\mathcal E$.
\begin{lemma}[Corollary~7 of \cite{ambainisqff}]
\label{cor:ff-classical}
    Let $P$ be any Markov chain on $X$, not necessarily reversible,
    and let $\rho$ be any initial distribution. Let
    $\mathcal E$ be the event that the original chain
    $Y=(Y_i)_{i\ge0}$, with $Y_0\sim\rho$, satisfies:
    \begin{enumerate}
        \item[\textup{(i)}] $Y_0\in U$;
        \item[\textup{(ii)}] a marked vertex is visited within
        the first $T$ steps;
        \item[\textup{(iii)}] at most $T$ of
        $Y_{T+1},\ldots,Y_{3T}$ are marked.
    \end{enumerate}
    Choose $r\in R$, $t\in\{1,\ldots,3T\}$ and
    $t'\in\{3T+1,\ldots,24T\}$ independently and uniformly,
    and set $s=1-\tfrac1r$. Couple $Y(s)$ to $Y$ by the
    geometric holding times described above. Then, whenever
    $\Pr(\mathcal E)>0$,
    \begin{equation}\label{eq:ff-cor11}
        \mathbb E_{t,t',s}\Bigl[
        \Pr\bigl(Y_0(s)\in U,\;
                 Y_t(s)\in M,\;
                 Y_{t'}(s)\in U
                 \,\big|\,\mathcal E\bigr)\Bigr]
        =
        \Omega\!\left(\frac{1}{\log T}\right).
    \end{equation}
\end{lemma}

The classical trajectory bound then transfers directly to the discriminant evolution of the chain, by removing the conditioning on $\mathcal E$ and averaging over $s$, $t$ and $t'$.

\begin{lemma}[cf. Corollary~8 of \cite{ambainisqff}]
\label{lem:ff-transfer}
    Let $P$ be a reversible ergodic Markov chain with stationary
    distribution $\pi$, and let $S$ be as above.
    Let $\beta>0$ and let $Y(s)$ evolve according to $P(s)$
    with $Y_0(s)\sim\pi$. If, for
    $t\in\{1,\ldots,3T\}$,
    $t'\in\{3T+1,\ldots,24T\}$ and $s\in S$
    chosen independently and uniformly,
    \begin{equation}\label{eq:ff-transfer-hyp}
        \mathbb E_{t,t',s}\Bigl[
        \Pr\bigl(Y_0(s)\in U,\;
                 Y_t(s)\in M,\;
                 Y_{t'}(s)\in U\bigr)\Bigr]
        =\Omega(\beta),
    \end{equation}
    then, for $s\in S$ and $t\in\{1,\ldots,24T\}$
    chosen independently and uniformly,
    \begin{equation}\label{eq:ff-transfer-concl}
        \mathbb E_{s,t}\Bigl[
        \bigl\lVert\Pi_M D(s)^t\ket{\pi_U}\bigr\rVert^2
        \Bigr]
        =\Omega(\beta).
    \end{equation}
\end{lemma}
\subsection{Truncated trajectory analysis}
\label{ssec:trunc-hitting}
In the original analysis of Ref.\cite{ambainisqff} the final time is chosen so that it satisfies
$T=\Omega(\HT(P,M))$, where a constant lower bound on  $\Pr(\mathcal E)$ follows immediately from Markov's inequality. For the
interpolated algorithm we must instead take $T\ll\HT(P,M)$, so a different argument
is required.

The following lemma provides the required estimate. It shows that, for
reversible Markov chains started from stationarity, the probability of hitting
the marked set within time $T$ is always at least linear in
$T/\HT(P,M)$ until it reaches a constant.

\begin{lemma}\label{lem:truncated-hitting}
	Let $P$ be an ergodic reversible Markov chain on $X$ with stationary
	distribution $\pi$, and let $\emptyset \neq M \subsetneq X$ be the set of
	marked elements. Let $Z$ denote the number of update steps until the walk
	first hits $M$, conditioned on the initial vertex being unmarked, so that
	$\mathbb{E}[Z] = \HT(P,M)$ by Definition~\ref{def:HT}. Then for every integer
	$t \ge 1$,
	\begin{equation}\label{eq:truncated-hitting}
		\Pr(Z \le t) \;\ge\; \frac{1}{14}\,
		\min\!\left(1,\; \frac{t}{\HT(P,M)}\right).
	\end{equation}
\end{lemma}

\begin{proof}
	Let $g(u) := \Pr(Z > u)$ denote the survival probability, so that $g$ is
	non-increasing with $g(0) = 1$. Since $Z$ is a non-negative integer random
	variable, summing its tail probabilities recovers its mean,
	\begin{equation}\label{eq:tail-sum}
		\sum_{u \ge 0} g(u)
		= \sum_{u \ge 0} \Pr(Z > u)
		= \sum_{u \ge 0} \sum_{z > u} \Pr(Z = z)
		= \sum_{z \ge 0} z\,\Pr(Z = z)
		= \mathbb{E}[Z] = \HT(P,M),
	\end{equation}
	where the interchange of summations is valid as all terms are non-negative.

	Consider $Y_0, Y_1, \ldots$ to be the chain started from $\pi$ conditioned on
	$Y_0 \in U$. Then $Z > t$ if and only if $Y_0, \ldots, Y_t \in U$, and
	since each transition within $U$ is governed by the substochastic block
	$P_{UU}$,
	\begin{equation}
		\label{eq:survival-killed}
		g(t) = \frac{1}{1-p_M}\sum_{x,y \in U} \pi_x\,(P_{UU}^{\,t})_{xy}.
	\end{equation}
	Furthermore, let $D_{UU}$ denote the $U \times U$
	block of the discriminant matrix. By detailed balance,
	$D_{xy} = \sqrt{P_{xy}P_{yx}} = \sqrt{\pi_x/\pi_y}\,P_{xy}$, so
	entrywise
	\begin{equation*}
		D_{UU} = \mathrm{diag}(\pi_U)^{1/2}\, P_{UU}\,
		\mathrm{diag}(\pi_U)^{-1/2},
	\end{equation*}
	hence $D_{UU}^{\,t} = \mathrm{diag}(\pi_U)^{1/2}\, P_{UU}^{\,t}\,\mathrm{diag}(\pi_U)^{-1/2}$. Substituting into
	\eqref{eq:survival-killed}, and writing $\sqrt{\pi_U}$ for the vector
	with entries $\sqrt{\pi_x}$, $x \in U$, we obtain
	\begin{equation}\label{eq:survival-spectral}
		g(t)
		= \frac{\sqrt{\pi_U}^{\,\top} D_{UU}^{\,t}\, \sqrt{\pi_U}}{1 - p_M}
		= \braket{U | D_{UU}^{\,t} | U}
		= \sum_{k} c_k\, \mu_k^{\,t},
		\qquad
		c_k := |\!\braket{v_k | U}\!|^2 \ge 0,
		\qquad
		\sum_k c_k = 1,
	\end{equation}
	where $\{\mu_k, \ket{v_k}\}$ is an orthonormal eigendecomposition of
	$D_{UU}$, since $P$ is reversible and therefore $D_{UU}$ is symmetric.
	Cauchy-Schwarz then implies
	\begin{align}\label{eq:doubling}
		g(u)^2 &=\Bigl(\sum_k c_k\mu_k^{\,u}\Bigr)^{\!2} \nonumber \\ \nonumber
		&=\Bigl(\sum_k \sqrt{c_k}\cdot\sqrt{c_k}\,\mu_k^{\,u}\Bigr)^{\!2}
		\\ \nonumber&\le \Bigl(\sum_k c_k\Bigr)\Bigl(\sum_k c_k\, \mu_k^{\,2u}\Bigr)
		\\ &= g(2u),
	\end{align}
	using $\sum_k c_k = 1$, valid for every integer $u$.
	Note that squaring leaves only the even powers $\mu_k^{\,2u} \ge 0$, and therefore non-negativity of the eigenvalues is not needed and no assumption on laziness is required.
	Iterating \eqref{eq:doubling} bounds $g$ at every time $2^j t$, with $j\ge0$. By
	induction on $j$, the case $j = 0$ being trivial, $g(2^{j+1}t) \ge
	g(2^j t)^2 \ge \bigl(g(t)^{2^j}\bigr)^2 = g(t)^{2^{j+1}}$, so that
	\begin{equation}\label{eq:dyadic-pin}
		g(2^j t) \;\ge\; g(t)^{2^j}
		\qquad\text{for every } j \ge 0.
	\end{equation}
	To exploit the doubling inequality, we estimate the tail sum at
	geometrically increasing times, which is sufficient to bound $g(t)$ through \eqref{eq:tail-sum}.
	Discarding the terms with $u < t$ and partitioning the rest into
	dyadic blocks $\{2^j t, \dots, 2^{j+1}t - 1\}$, each containing $2^j t$
	indices on which $g$ is at least its right-endpoint value $g(2^{j+1}t)$ (since it is non-increasing),
	\begin{align}\label{eq:dyadic}
		\HT(P,M) &= \sum_{u \ge 0} g(u)
		\ge\;\sum_{u\geq t}g(u) \quad\quad,\forall t\in\mathbb N\nonumber
		\\&=\; \sum_{j\geq 0}\sum_{u=2^jt}^{2^{j+1}t-1}g(u) \nonumber
		\\&\ge\; \sum_{j \ge 0} 2^j t\, g(2^{j+1} t) \nonumber
		\\&\ge\; t \sum_{j \ge 0} 2^j\, g(t)^{\,2^{j+1}}.
	\end{align}
	Now, set $L := \ln(1/g(t)) > 0$. If $L \ge \tfrac14$, then $g(t) = e^{-L} \le e^{-1/4}$
	and the bound is trivial
	\begin{equation*}
		\Pr(Z \le t) = 1 - g(t) \;\ge\; 1 - e^{-1/4}
		\;>\; \tfrac15 \;\ge\; \tfrac{1}{14}\min\!\bigl(1, t/\HT(P,M)\bigr).
	\end{equation*}
	Assume henceforth $L < \tfrac14$. Each term of the sum in
	\eqref{eq:dyadic} can be written, with $n := 2^{j+1}$, as
	\begin{equation*}
		t\, 2^j\, g(t)^{\,2^{j+1}} = \tfrac{t}{2}\, n\, e^{-Ln}.
	\end{equation*}
	Since the map $n \mapsto n\,e^{-Ln}$ peaks at $n = 1/L$, it suffices to
	retain the single term whose index is a power of two nearest that peak. Because its endpoints differ by a factor of two and
	since $L < \tfrac14$, the interval
	$[\tfrac{1}{2L}, \tfrac1L]$ contains some $n = 2^{j+1}$ with $j \ge 0$. Furthermore, this
	$n$ satisfies $\tfrac{1}{2L}\leq n \leq \tfrac1L$, so retaining only this term of
	\eqref{eq:dyadic} gives the bound
	\begin{equation*}
		\HT(P,M) \;\ge\; \tfrac{t}{2}\, n\, e^{-Ln}
		\;\ge\;\frac{t}{4eL}.
	\end{equation*}
	Finally, the elementary inequality $\ln(1/x) \le (1-x)/x$ on $(0,1]$ gives
	$L \le (1 - g(t))/g(t) \le e^{1/4}(1 - g(t))$. Substituting into $\HT(P,M) \ge t/(4eL)$ and
	rearranging,
	\begin{equation}
		\Pr(Z \le t) = 1 - g(t)
		\;\ge\; \frac{t}{4 e^{5/4}\,\HT(P,M)}
		\;\ge\; \frac{1}{14}\cdot\frac{t}{\HT(P,M)},
	\end{equation}
	since $4 e^{5/4} \le 14$. Combining the two cases, and noting that the
	right-hand side of \eqref{eq:truncated-hitting} is capped at $\tfrac{1}{14}$
	once $t \ge \HT(P,M)$, gives the desired result.
	\end{proof}

This result allows the trajectory analysis of Ref.\cite{ambainisqff} to be extended to arbitrary times $T \ge 1$, with the expected loss of a
factor $\min(1,T/\HT(P,M))$ in the per-run success probability.

\begin{lemma}\label{cor:qff-trunc-perrun}
	Let $P$ be a reversible ergodic Markov chain with stationary distribution
	$\pi$ and let $T \ge 1$ be an integer. Suppose
	\begin{equation}\label{eq:qff-trunc-pm}
		p_M \;\le\; \frac{1}{112}\,\min\!\Biggl(1,\frac{T}{\HT(P,M)}\Biggr).
	\end{equation}
	Then, choosing $s \in S$ and $t \in \{1,\dots,24T\}$ independently and
	uniformly at random,
	\begin{equation}\label{eq:qff-trunc-perrun}
		\mathbb{E}_{s,t}\Bigl[\bigl\lVert \Pi_M D(s)^t
		\ket{\pi_U}\bigr\rVert^2\Bigr]
		\;=\; \Omega\!\left(\frac{1}{\log T}\,
		\min\!\Biggl(1,\frac{T}{\HT(P,M)}\Biggr)\right).
	\end{equation}
\end{lemma}

\begin{proof}
	Consider the original chain $(Y_u)_{u\ge0}$ started from
$\pi$ and without conditioning on $Y_0\in U$.
	Let
	$\tau_M := \min\{u \ge 0 : Y_u \in M\}$ and recall the event $\mathcal E$ of
	Lemma~\ref{cor:ff-classical} defined by the following conditions: (i)~$Y_0 \in U$, (ii)~$M$ is hit within
	the first $T$ steps ($\tau_M\le T$) and (iii)~at most $T$ of the $2T$ steps in the window
	$[T,3T]$ are marked. Let $Z$ be the hitting time of $M$ from $\pi$
	conditioned on $Y_0 \in U$, so $\mathbb{E}[Z] = \HT(P,M)$ by
	Definition~\ref{def:HT}.

	Furthermore, note that by assumption $p_M\le 1/112<1/2$, so by Lemma~\ref{lem:truncated-hitting}
	\begin{equation}\label{eq:qff-trunc-12}
		\Pr(Y_0 \in U,\ \tau_M \le T)
		= (1-p_M)\,\Pr(Z \le T)
		\;\ge\; \tfrac12\cdot\tfrac{1}{14}\,
		\min\!\Biggl(1,\frac{T}{\HT(P,M)}\Biggr).
	\end{equation}
	It remains to include condition (iii), which requires at
most $T$ visits to the marked set during the interval $(T,3T]$.
	By stationarity we have
	$\mathbb{E}\Bigl[
  \bigl|\{i\in(T,3T]:Y_i\in M\}\bigr|
\Bigr]=2Tp_M$, so using Markov's
	inequality we get that $\Pr(\neg(\mathrm{iii})) \le 2p_M$. Using the elementary bound
	$\Pr(A\wedge B\wedge C) \ge \Pr(A\wedge B) - \Pr(\neg C)$ together with
	\eqref{eq:qff-trunc-12} and \eqref{eq:qff-trunc-pm} yields
	\begin{align}\label{eq:qff-trunc-E}
		\Pr(\mathcal E)
		&\ge\; \frac{1}{28}\min\!\Biggl(1,\frac{T}{\HT(P,M)}\Biggr)
		- 2p_M\ge\; \frac{1}{56}\min\!\Biggl(1,\frac{T}{\HT(P,M)}\Biggr).
	\end{align}
	Since $\mathcal E$ depends only on the base trajectory and not on
	$t,t',s$, removing the conditioning and applying
	Lemma~\ref{cor:ff-classical},
	\begin{align}
		\mathbb{E}_{t,t',s}\bigl[\Pr(Y_0(s)\in U,\,Y_t(s)\in M,\,
		Y_{t'}(s)\in U)\bigr]
		&\ge\; \Pr(\mathcal E)\cdot\;\Omega\!\left(\frac{1}{\log T}\right) \nonumber
		\\&=\; \Omega\!\left(\frac{1}{\log T}\min\!\biggl(1,\frac{T}{\HT(P,M)}\biggr)\right),
	\end{align}
	with $t \in \{1,\dots,3T\}$, $t' \in \{3T+1,\dots,24T\}$, $s\in S$ uniform.
	Applying Lemma~\ref{lem:ff-transfer} with
	$\beta = \tfrac{1}{\log T}\min(1,T/\HT(P,M))$ yields \eqref{eq:qff-trunc-perrun}.
\end{proof}

\subsection{The interpolation theorem}\label{ssec:qff-interp-thm}

We now state the algorithm and interpolation result.

\bigskip

\begin{algorithm}[H]
\caption{\textsc{InterpolatedAGJK}$(P, M, H, \gamma)$, for $\gamma\in[0,1]$}\label{alg:agjk}

\noindent Set $T_\gamma := \lceil H^{\gamma}\rceil$ and
$\varepsilon' := \Theta\!\bigl(H^{(\gamma-1)/2}/\sqrt{1+\gamma\log H}\bigr)$.

\begin{enumerate}
	\item Using $\Setup(P)$ and uniform preparations on the time and scale
	registers, prepare
	\[
	\frac{1}{\sqrt{24T_\gamma}}\sum_{t=1}^{24T_\gamma}\ket{t}_{\mathsf R_t}
	\otimes
	\frac{1}{\sqrt{|S|}}\sum_{s\in S}\ket{s}_{\mathsf R_s}
	\otimes
	\ket{\pi}_{\mathsf R_1}\ket{\bar0}_{\mathsf R_2}.
	\]
	\item Attach a check ancilla, apply $\Checkk(M)$ to the vertex register
	$\mathsf R_1$, and measure it. If the outcome is marked, measure
	$\mathsf R_1$ in the vertex basis and output the result. Otherwise
	$\mathsf R_1$ holds the subnormalised state $\ket{\pi_U}$.
	\item Controlled on $\mathsf R_t\mathsf R_s$, apply
\QFF$\bigl(\ket{\pi_U},P(s),t,\varepsilon'\bigr)$,
using polynomial precision $\varepsilon'$, whose successful branch before normalisation maps $\ket t\ket s\ket{\pi_U}\ket{\bar0}$ to
\[
\ket t\ket s\,
g_{t,\tau}(D(s))\ket{\pi_U}\ket{\bar0},
\]
where $g_{t,\tau}$ is the truncated
	polynomial applied by QFF.
If QFF succeeds, measure the vertex register and apply
$\Checkk(M)$; if the vertex is marked, output it and stop.
Otherwise, output \texttt{No marked vertex found}.
\end{enumerate}
\end{algorithm}
\bigskip

\begin{namedresult}{Theorem~\ref{thm:interp-qff}}

	Let $P$ be a reversible ergodic Markov chain on a finite state space $X$,
	let $M \subseteq X$ be non-empty, and suppose a bound $\HT \ge \HT(P,M)$ is
	known. For every $\gamma \in [0,1]$,
	there is a quantum walk-based algorithm of maximum walk depth
	\begin{equation}
		D(\gamma) \;=\; O\!\left(\HT^{\gamma/2}\,
		\sqrt{(1-\gamma)\log \HT + \log\log \HT}\right),
	\end{equation}
	which outputs a marked vertex with constant success probability and total complexity of order at most
	\begin{equation}
		(1+\gamma\log \HT)\,\HT^{1-\gamma}\cdot\mathsf S
		\;+\;
		(1+\gamma\log \HT)\,\HT^{1-\gamma/2}\,
		\sqrt{(1-\gamma)\log \HT + \log\log \HT}\,\cdot(\mathsf U+\mathsf C)\,.
	\end{equation}
\end{namedresult}

\begin{proof}
	Throughout the proof we abbreviate the given upper bound as $H := \HT$. Consider the routine \IAGJK$(P, M, H, \gamma)$, which sets $T \leftarrow T_\gamma := \lceil H^{\gamma}\rceil$, so $1 \le T \le H + 1$ and
	$\min(1, T/\HT(P,M)) \ge H^{\gamma-1}$, since $\HT(P,M) \le H$.

	Suppose $p_M \le \tfrac{1}{112}\min(1, T/\HT(P,M))= \Omega(H^{\gamma-1})$ then, by Lemma~\ref{cor:qff-trunc-perrun}, choosing $s\in S$ and $t\in\{1,...,24T_\gamma\}$ uniformly at random, we have
	\begin{equation}\label{eq:qgamma-def}
		\mathbb{E}_{s,t}\Bigl[\bigl\lVert \Pi_M D(s)^t
		\ket{\pi_U}\bigr\rVert^2\Bigr] = \Omega\!\left(\frac{1}{\log T_ \gamma}\,H^{\gamma-1}\right).
	\end{equation}
	Let $g_{t,\tau}(D(s))$ denote the QSVT polynomial implemented in step 3 approximating $D(s)^t$ with degree $\tau=\ceil{\sqrt{2t\ln(2/\varepsilon')}}$, where $\varepsilon'$ is the operator-norm approximation error to be chosen and defined by
	\begin{equation}
		\lVert g_{t,\tau}(D(s)) - D(s)^t\rVert \le \varepsilon'.
	\end{equation}
	Then, by the reverse triangle inequality
	\begin{align}
		\bigg|\left\lVert \Pi_M g_{t,\tau}(D(s))\ket{\pi_U}\right\rVert - \left\lVert\Pi_M D(s)^t\ket{\pi_U}\right\rVert\bigg|&\le\;\left\lVert \left(\Pi_M g_{t,\tau}(D(s))-\Pi_M D(s)^t\right)\ket{\pi_U}\right\rVert\nonumber\\&
		\le\;\left\lVert \Pi_M g_{t,\tau}(D(s))-\Pi_M D(s)^t\right\rVert\left\lVert\ket{\pi_U}\right\rVert\nonumber\\\nonumber&\le\;\left\lVert g_{t,\tau}(D(s)) - D(s)^t\right\rVert \\&\le\; \varepsilon',
	\end{align}
	since projection can only shrink the operator norm and $\left\lVert\ket{\pi_U}\right\rVert\le1$.
	Therefore,
	\begin{align}
		\left\lVert \Pi_M g_{t,\tau}(D(s))\ket{\pi_U}\right\rVert^2&\ge\;\left\lVert \Pi_M D(s)^t\ket{\pi_U}\right\rVert^2-2\left\lVert \Pi_M D(s)^t\ket{\pi_U}\right\rVert\,\varepsilon'.
	\end{align}
	Taking the uniform expectation over $(s,t)$ and using Jensen's inequality gives the lower bound
	\begin{align}\label{eq:succpqff}
		\mathbb{E}_{s,t}\Bigl[\bigl\lVert \Pi_M g_{t,\tau}(D(s))
		\ket{\pi_U}\bigr\rVert^2\Bigr]&\ge\;\mathbb{E}_{s,t}\Bigl[\bigl\lVert \Pi_M D(s)^t
		\ket{\pi_U}\bigr\rVert^2-2\varepsilon'\bigl\lVert \Pi_M D(s)^t
		\ket{\pi_U}\bigr\rVert\Bigr]\nonumber
		\\&\ge\;\mathbb{E}_{s,t}\Bigl[\bigl\lVert \Pi_M D(s)^t
		\ket{\pi_U}\bigr\rVert^2\Bigr]-2\varepsilon'\mathbb{E}_{s,t}\Bigl[\bigl\lVert \Pi_M D(s)^t
		\ket{\pi_U}\bigr\rVert\Bigr]
	\end{align}
	Then it is sufficient to take the operator-norm precision to be of order
\begin{align}
	\label{eq:qffsearchprecision}
	\varepsilon' =\; \Theta\left(\frac{H^{(\gamma-1)/2}}{\sqrt{\log T_ \gamma}}\right) =\;
	\Theta\left(\frac{H^{(\gamma-1)/2}}{\sqrt{1+\gamma\log H}}\right)\,,
\end{align}
where the additive $1$ factor was included to prevent divergence at $\gamma=0$ whereas it was initially omitted in Lemma~\ref{cor:ff-classical} since $|S|=1+\log T_\gamma=O(\log T_\gamma)$.

Under the choice of precision from \eqref{eq:qffsearchprecision}, the QSVT polynomial approximating the fast-forwarded Markov chain $D(s)^t$  has degree of order
\begin{align}
	\tau &=\; O\left(\sqrt{T_\gamma\log(2/\varepsilon')}\right)
	\nonumber \\ &= \; O \left(H^{\gamma/2}\sqrt{(1-\gamma)\log H+\log(1+\gamma\log H)}\right) \nonumber \\&=\;
	 O \left(H^{\gamma/2}\sqrt{(1-\gamma)\log H+\log\log H}\right)\,.
\end{align}
The QSVT implementation of $g_{t,\tau}(D(s))
$ uses $O(\tau)$ applications of the operator $U_{D}(s)$	since the implemented polynomial degree exactly matches the QSVT circuit depth. So the depth is
\begin{equation}
	D(\gamma) = \, O \left(H^{\gamma/2}\sqrt{(1-\gamma)\log H+\log\log H}\right)\,.
\end{equation}
The success probability is directly given by \eqref{eq:succpqff} under the optimal choice of $\varepsilon'$, which equates to \eqref{eq:qgamma-def}. Therefore, repeating the algorithm a total of
\begin{align}
	R(\gamma)= O\left(\left(1+\gamma\log H\right)H^{1-\gamma}\right)
\end{align}
yields a marked vertex with constant success probability.
Since each application of $U_{D}(s)$ has cost $O(\mathsf U + \mathsf C)$, then the total cost is of order at most
\begin{align}
	\label{eq:complqffint}
	R(\gamma)\cdot\left(\mathsf S + D(\gamma)\cdot\left(\mathsf U +\mathsf C\right)\right) &= (1+\gamma\log H)H^{1-\gamma}\cdot\mathsf S\nonumber + (1+\gamma\log H)H^{1-\gamma}\\
	&\quad
	\cdot H^{\gamma/2}\sqrt{(1-\gamma)\log H + \log\log H}
	\cdot(\mathsf U+\mathsf C)\nonumber \\
	&= (1+\gamma\log H)H^{1-\gamma}\cdot\mathsf S\nonumber \\
	&\quad + (1+\gamma\log H)H^{1-\gamma/2}
	\sqrt{(1-\gamma)\log H + \log\log H}\cdot(\mathsf U+\mathsf C).
\end{align}
If $p_M > \tfrac{1}{112}\min(1, T/\HT(P,M)) \geq \tfrac{1}{112}H^{\gamma-1}$, the initial $\Checkk(M)$ on
$\ket{\pi}$ (step 2) returns a marked vertex with probability
$p_M = \Omega(H^{\gamma-1})$ and hence a marked vertex is output with constant success probability after $O\left(H^{1-\gamma}\right)=O(R(\gamma))$ repetitions, which is absorbed into the bound \eqref{eq:complqffint}.
\end{proof}

\section{Discussion}\label{Discussion}

In this work we have characterized the depth-complexity trade-offs of quantum-walk spatial search.
Our results extend the programme of quantum-classical interpolation, which previously included unstructured search~\cite{sun2019hybrid}, amplitude and phase estimation~\cite{wang2019accelerated,giurgica2022qjournal,Rall2023amplitudeestimation,magano2022simplifying}, and AND-OR trees~\cite{murca2024making}.
Interpolating the KMOR~\cite{kmor} and AGJK~\cite{ambainisqff} algorithms, we obtain families parameterised by $\gamma\in[0,1]$ in which, writing $H$ for the relevant hitting time ($\HTp$ for KMOR and $\HT$ for AGJK), a marked vertex is found with constant probability with cost
\begin{equation}
    \widetilde O\!\left(
        \frac{H}{D^2}\cdot\mathsf S
        \;+\;
        \frac{H}{D}\cdot(\mathsf U+\mathsf C)
    \right),
\end{equation}
without using circuits of depth greater than $D$.

The dependence on $\mathsf C$ is optimal (up to logarithmic factors) in the worst case over graphs, as spatial search on the complete graph recovers unstructured search \cite{sun2019hybrid}. However, this only holds for bounds concerning the hitting time. On particular graphs fewer checks can suffice by quantising classical random walk algorithms
that check for marked vertices only every few steps, such as those underlying the
MNRS framework~\cite{magniez2011search} and controlled quantum
amplification~\cite{dohotaru2017controlled}, both of which are recovered by the
unified framework of Ref.~\cite{apers2021unified}.
Furthermore, this lower bound does not constrain the dependence on
$\mathsf S$ or $\mathsf U$. To the best of our knowledge, no
corresponding lower bounds are known in a model that accounts
separately for setup, update, and checking costs. Determining whether
these dependencies are optimal, or can be improved by different
algorithmic approaches, remains an open problem.

The difference in the relevant hitting times reflects the current understanding of quantum walk search with multiple marked vertices. Based on numerical evidence, Ref.\cite{ambainisqff} conjectured that a suitable instance dependent choice of the randomised parameters always suffices to find a marked vertex in $O(\sqrt{\HT})$ steps, which to our knowledge remains unresolved. A proof would be a natural starting point for an interpolation in terms of $\HT$ free of logarithmic factors, for any number of marked vertices.

Despite these differences, both families share a limitation that is invisible in previous interpolations. Both algorithms reduce the number of walk steps in terms of the hitting time, which is defined for a walk started from the stationary distribution.
This means that, in this framework, any (quantum) attempt at finding a marked vertex, be it at full or limited-depth, must begin by performing some initial state preparation routine to produce a quantum state encoding the stationary distribution. Reducing the number of coherent walk steps hereby results in a decrease in the success probability of each run, which must be compensated by taking more classical samples through fresh invocations of the search algorithm, leading to the stationary distribution being prepared just as many times. This creates a dichotomy where the usefulness of shorter depth search rests on the relation between the cost of initialising the system to the stationary distribution and that of advancing the walk by one step.

Naturally, one might consider an alternative approach to mitigate the setup overhead, where, rather than restarting the algorithm from scratch after a failed measurement, we simply recycle the post-measurement state indefinitely.
While this strategy could successfully reduce the state preparation cost to $O(1)$, it pays the price in terms of coherence requirements, fundamentally challenging the core motivation behind quantum-classical interpolation in the first place. What state recycling achieves is to push all $\widetilde O(H^{1-\gamma})$ independent repetitions of depth $\widetilde O(H^{\gamma/2})$ to a single extended coherent run of depth $\widetilde O(H^{1-\gamma/2})$.
Consequently, at the classical limit, this approach simply produces a quantum analogue of the classical random walk search where the walk is interleaved with partial measurements in between every step, although coherently. This heavily penalises the algorithm with high coherence requirements and strictly worse complexity compared to the fully quantum baseline. Altogether, this suggests that constructing a family of algorithms capable of perfectly achieving both limits requires a fundamental departure from these approaches. Most notably, it demands eliminating the reliance on initial state preparations and, consequently, departing from the traditional hitting time framework, which is dependent on the stationary distribution.

The electric network framework of Belovs~\cite{belovs2013quantum} offers such a departure, allowing the walk to start from an arbitrary initial distribution, with the hitting time replaced by the effective resistance between its support and the marked set, multiplied by the total weight of the graph. This quantity generalises the commute time, to which it reduces, up to a constant
factor, when starting from a single vertex~\cite{chandra,apers2021unified}. Originally formulated for detection, it was extended to finding by Piddock~\cite{piddock2019} and, using the QFF primitive, by Apers, Gilyén and Jeffery~\cite{apers2021unified}, up to logarithmic factors. While starting from a single vertex makes the initial state preparation essentially free, the effective resistance to the marked set can be substantially larger than the hitting time, essentially trading the setup penalty for a potentially longer walk. As the hitting time setting is recovered when starting from the stationary
distribution, an interpolation in this framework would generalise the one presented in this work, where the optimal choice of initial distribution is dictated by graph topology and hardware costs.

Whether a depth-limited family of algorithms can be constructed within this framework to seamlessly interpolate between the classical and fully coherent limits (possibly evading the bottlenecks of the hitting time setting) stands as a compelling open problem for future work. More broadly, the interpolation of other quantum walk-based algorithms, namely element distinctness~\cite{ambainis2007quantum} (for which an upper bound is noted in \cite{ambainis2026tight}), triangle finding~\cite{magniez2007triangle}, evaluation of AND-OR formulas~\cite{farhi2008nand,ambainis2010andor} and the quadratic acceleration of Markov chain Monte Carlo methods, such as simulated annealing and the estimation of partition functions~\cite{somma2008annealing,harrow2020adaptive,cornelissen2023sublinear} also constitutes a natural direction.

We also provide an interpolation of fixed-point amplitude amplification that is worst-case optimal at constant failure probability. Formulated in the same language of parameterised families, this result can be either applied directly or adapted in broader scenarios for designing more complex interpolation protocols. Furthermore, all new results are framed in the language of the QSVT formalism, connecting circuit depth to the polynomial degree of the underlying approximation and extending the line of thought of Refs.~\cite{magano2022simplifying,murca2024making} showing that QSVT is not only a unifying framework for quantum algorithms but also offers a convenient lens for hybrid protocol design.

\section*{Acknowledgements}
DM acknowledges support from FCT-Portugal through Grant No. UID/04650/2025.
DR acknowledges support from FCT-Portugal through LaPMET (LA/P/0095/2020), grant Ref.~26/ECUM/CFUM/2026.

\section*{Author contributions}
DM first proposed the problem of interpolating spatial search with quantum walks.
DR developed the main proofs.
All authors contributed to the writing and revision of the manuscript.

\section*{AI disclosure}

The authors developed the key ideas and wrote the manuscript.
The authors acknowledge the usage of AI to review the writing, help write the code used to generate the figures, assist finding and editing references, and for minor technical assistance when verifying the proofs.
All contents were reviewed and verified independently and the authors take full responsibility for the contents of the manuscript.

\newpage
\appendix

\section{Limited knowledge variations}\label{variations}

Theorems~\ref{thm:cq_interp} and~\ref{thm:interp-qff} assume knowledge of $p_M$ and $\HT(P,M)$, and of an upper bound on $\HT(P,M)$, respectively. We now show, generally following the full-depth approaches of Refs.~\cite{kmor,ambainisqff}, that these requirements can be substantially relaxed for every $\gamma\in[0,1]$.

We begin by considering settings in which $p_M$ or the extended hitting time $\HTp(P,M)$ is known only approximately, is merely bounded, or is unknown for the KMOR interpolation. The $\HT$ dependence follows by specialising the general forms below to the single marked vertex case.

Furthermore, all outer searches over candidate parameters are performed classically, to keep the depth minimal. If preferable, those of Corollaries~\ref{thm:cq_bound_pm} and~\ref{thm:cq_bound_H}, Theorem~\ref{thm:interp-qff} and Corollary~\ref{cor:interp-qff-unknown} can be replaced by coherent amplitude amplification, or alternatively by its interpolated version (Theorem~\ref{thm:interp-aa}) as previously commented.

\subsection{Variations of the KMOR algorithm}
\subsubsection{\texorpdfstring{Approximate $p_M$ known and $\HTp$ known}{Approximate pM known and HTp known}}

We first evaluate whether approximate knowledge of $p_M$ suffices to verify the same conditions as the general form of Theorem~\ref{thm:cq_interp}.
A key structural difference from the $\gamma = 1$ case is that the
interpolation parameter $s_\gamma$, defined implicitly by
the choice $\sin^2\!\theta(s_\gamma) = (\HTp)^{\gamma-1}$, takes the explicit form
\begin{equation}\label{eq:sgamma_explicit}
	s_\gamma \;=\;
	\frac{(\HTp)^{\gamma-1} - p_M}{(\HTp)^{\gamma-1}(1-p_M)},
\end{equation}
which depends jointly on $p_M$ and $\HTp$.

When the value of $(\HTp)^{\gamma-1}$ is available and we are given an approximate value of $p_M$, the
natural extension is to consider an approximate interpolation parameter obtained by substituting the
approximation $p^*$ for $p_M$
\begin{equation}\label{eq:sgamma_star_def}
	\tilde{s}_\gamma \;:=\;
	\frac{(\HTp)^{\gamma-1} - p^*}{(\HTp)^{\gamma-1}(1-p^*)}.
\end{equation}
The following lemma bounds the overlap
$\sin^2\!\theta(\tilde{s}_\gamma)$ under the true stationary probability
$p_M$ when only the approximation $p^*$ is used to compute
$\tilde{s}_\gamma$.

\begin{lemma}\label{lem:overlap_bounds}
	Let $p_M \leq 1/2$ and assume we are given an approximation $p^*$ such that
	$|p^* - p_M| \leq p_M/3$. For $\eta \in (0,1)$ satisfying
	$p^* \leq \tfrac{1}{2}\eta$, set
	\begin{equation}\label{eq:stilde_def}
		\tilde{s} \;:=\;
		\frac{\eta - p^*}{\eta(1-p^*)}.
	\end{equation}
	Then $\tilde{s} \in [0,1)$ and
	\begin{equation}\label{eq:overlap_bounds}
		\tfrac{1}{2}\,\eta
		\;\leq\; \sin^2\!\theta(\tilde{s})
		\;\leq\; 2\,\eta.
	\end{equation}
\end{lemma}
\begin{proof}
	First, we verify that $\tilde{s}$ falls within the required domain. By assumption, $p^* \leq \tfrac{1}{2}\eta < \eta$, which makes the numerator $\eta - p^*$ strictly positive. Since $p^*$ and $\eta$ are both positive, the denominator $\eta(1-p^*)$ is also positive and $
\eta(1-p^*)-(\eta-p^*)
=
p^*(1-\eta)>0,$ thus $\tilde{s}\in[0,1)$.
Moreover, $\tilde s$ is strictly decreasing in $p^*$, as
	\begin{equation}
		\frac{d\tilde{s}}{dp^*}
		=
		\frac{\eta - 1}{(1-p^*)^2\,\eta}<0,
	\end{equation}
	while by Proposition~\ref{prop:principal}, $\sin^2\!\theta(s)=p_M/(1-s(1-p_M))$ is strictly increasing in $s$. Hence
	\begin{equation}\label{eq:sin2_general_revised}
		\sin^2\!\theta(\tilde{s})
		=
		\frac{p_M(1-p^*)\,\eta}
		{p^*(1-p_M) + (p_M - p^*)\,\eta}
	\end{equation}
	is decreasing in $p^*$ and since the condition $|p^* - p_M| \leq p_M/3$ restricts $p^*$ to the interval
	$
	\left[\tfrac{2}{3}p_M,\tfrac{4}{3}p_M\right],
	$
	the composition $\sin^2\!\theta(\tilde{s})$ is also monotonically decreasing in $p^*$ so it remains to evaluate it at the two endpoints.

	Substituting and using $p_M\le\tfrac12$,
	\[
		\sin^2\!\theta(\tilde{s})
		\le\frac{\eta\,(3 - 2p_M)}{2(1-p_M) + \eta}
		<\frac{\eta\,(3 - 2p_M)}{2(1-p_M)}
		\le 2\eta,
		\qquad
		\sin^2\!\theta(\tilde{s})
		\ge\frac{\eta\,(3 - 4p_M)}{4(1-p_M) - \eta}
		>\frac{\eta\,(3 - 4p_M)}{4(1-p_M)}
		\ge \tfrac12\eta,
	\]
	where the last inequalities follow since the first bound is increasing and the second decreasing in $p_M$ on $[0,\tfrac12]$.
\end{proof}

Applying this bound with $\eta = (\HTp)^{\gamma-1}$ preserves the required scaling of the overlap under the perturbed parameter $\tilde{s}_\gamma$, retaining the complexity of the perfect knowledge interpolation up to at most constant factors.

\begin{corollary}\label{thm:cq_approx_pm}
	Let $P$ be an ergodic, reversible Markov chain on $G=(X,E)$ and let
	$M\subseteq X$ be non-empty with $p_M\le\tfrac12$.
	Suppose an approximation $p^*$ of $p_M$ satisfying
	$|p^* - p_M| \leq p_M/3$ is given and $\HTp$ is known. Then, for any
	$\gamma\in[0,1]$, there is a quantum walk-based algorithm with
	quantum-walk depth
	\[
		D(\gamma)=O\!\left((\HTp)^{\gamma/2}\right)
	\]
	which finds a marked vertex with constant success probability and
	expected hybrid complexity of order
	\begin{equation}\label{eq:cq_approx_pm_complexity}
		(\HTp)^{1-\gamma}\cdot \mathsf{S}
		\;+\; (\HTp)^{1-\gamma/2}\cdot(\mathsf{U}+\mathsf{C}).
	\end{equation}
\end{corollary}

\begin{proof}[Proof of Corollary~\ref{thm:cq_approx_pm}]
	Throughout the proof we abbreviate $H := \HTp$, and take
	\begin{equation}\label{eq:sgamma_star}
		\tilde{s}_\gamma
		\;:=\; \frac{H^{\gamma-1} - p^*}{H^{\gamma-1}(1-p^*)},
		\qquad
		\varepsilon_3 \;:=\; \frac{\sqrt2}{24},
		\qquad
		t_\gamma \;:=\; 72\,\ceil{H^{\gamma}}\,,
	\end{equation}
	when possible.
	Furthermore, note that, by assumption, $\tfrac{2}{3}p_M \leq p^* \leq \tfrac{4}{3}p_M$.

	Suppose first that $p^* > \tfrac{1}{2}H^{\gamma-1}$.  Then
	$p_M \geq \tfrac{3}{4}p^* > \tfrac{3}{8}H^{\gamma-1}$, so each
	independent sample from $\pi$ via $\Setup(P)$ lands on a marked vertex
	with probability $\Omega(H^{\gamma-1})$.  After $O(H^{1-\gamma})$
	repetitions, each costing $\mathsf{S} + \mathsf{C}$, a marked vertex is
	found with constant probability.  The total cost
	$O(H^{1-\gamma})\cdot(\mathsf{S}+\mathsf{C})$ is absorbed
	into~\eqref{eq:cq_approx_pm_complexity} since
	$H^{1-\gamma} \leq H^{1-\gamma/2}$.

	Suppose instead that $p^* \leq \tfrac{1}{2}H^{\gamma-1}$, and consider
	first the sub-regime $H^{\gamma-1} \geq \tfrac{1}{6}$.  Then
	$H^{1-\gamma} \leq 6$, so $O(1)$ calls to the full-depth algorithm
	suffice. By Theorem~7 of Ref.\cite{kmor}, we
	can find a marked vertex with constant probability and complexity
	$O(\mathsf{S} + \sqrt{H}\,(\mathsf{U}+\mathsf{C}))$.  Since
	$H^{1-\gamma} = O(1)$ and
	$H^{1-\gamma/2} \geq H^{1/2}/\sqrt{6} = \Omega(\sqrt{H})$ in this
	regime, the cost is absorbed into~\eqref{eq:cq_approx_pm_complexity}.

	It remains to treat the regime
	$p^* \leq \tfrac{1}{2}H^{\gamma-1}$ with $H^{\gamma-1} < \tfrac{1}{6}$.
	Since $H^{\gamma-1} < 1$, all hypotheses of
	Lemma~\ref{lem:overlap_bounds} are satisfied with $\eta = H^{\gamma-1}$,
	giving
	\begin{equation}\label{eq:sin2_scenario1}
		\tfrac{1}{2}\,H^{\gamma-1}
		\;\leq\; \sin^2\!\theta(\tilde{s}_\gamma)
		\;\leq\; 2\,H^{\gamma-1}\,,
	\end{equation}
	which by the upper bound and $H^{\gamma-1} < 1/6$ yields
	$\sin^2\!\theta(\tilde{s}_\gamma) < 1/3$ and
	$\cos^2\!\theta(\tilde{s}_\gamma) > 2/3$.  Combining with the lower
	bound we have
	\begin{equation}\label{eq:eps1_scenario1}
		\varepsilon_1
		\;:=\; \cos\theta(\tilde{s}_\gamma)\,\sin\theta(\tilde{s}_\gamma)
		\;>\; \sqrt{\tfrac{2}{3}}
		\cdot \sqrt{\tfrac{1}{2}\,H^{\gamma-1}}
		\;=\; \tfrac{1}{\sqrt{3}}\,H^{(\gamma-1)/2}
		\;>\; \tfrac{1}{2}\,H^{(\gamma-1)/2}.
	\end{equation}
	By Theorem~\ref{thm:HT-s-relation} and the upper bound
	in~\eqref{eq:sin2_scenario1}
	\begin{equation}\label{eq:HT_scenario1}
		\HT(\tilde{s}_\gamma)
		\;=\; \sin^4\!\theta(\tilde{s}_\gamma)\cdot H
		\;\leq\; 4\,H^{2\gamma-1}\,
	\end{equation}
	so by setting
	\begin{equation}\label{eq:eps2_scenario1}
		\varepsilon_2 \;:=\; \tfrac{1}{6}\,H^{(\gamma-1)/2},
	\end{equation}
	the fast-forwarding time condition of Theorem~\ref{thm:main-qff} is
	verified by direct computation,
	\begin{equation}\label{eq:depth_scenario1}
		\frac{\HT(\tilde{s}_\gamma)}{2\,\varepsilon_2^{2}}
		\;\leq\;
		\frac{4\,H^{2\gamma-1}}{2}\cdot\frac{36}{H^{\gamma-1}}
		\;=\; 72\,H^{\gamma}
		\;=\; t_\gamma ,
	\end{equation}
	and by \eqref{eq:T-qff} the resulting coherent depth is
	\begin{equation}\label{eq:Tgamma_scenario1}
		T_\gamma
		\;=\; \Bigl\lceil \sqrt{2\,t_\gamma\ln(2/\varepsilon_3)} \Bigr\rceil
		\;=\; \Bigl\lceil \sqrt{144\,\ln(24\sqrt2)}\;H^{\gamma/2}\Bigr\rceil
		\;=\; O\!\bigl(H^{\gamma/2}\bigr),
	\end{equation}
	as claimed.  The ordering $\varepsilon_1 > \varepsilon_2$ follows trivially from
	$1/2 > 1/6$, and $\varepsilon_1 \leq 1/2$ since
	$\cos\theta\sin\theta = \tfrac{1}{2}\sin 2\theta \leq 1/2$. Moreover, $0 < \varepsilon_2 + 2\varepsilon_3\sin\theta(\tilde{s}_\gamma)
\leq \tfrac{1}{3}H^{(\gamma-1)/2} < \varepsilon_1$.

	All hypotheses of Theorem~\ref{thm:main-qff} are therefore satisfied and
	 each call to
	\KMOR$(P,M,\tilde{s}_\gamma,t_\gamma,\varepsilon_3)$ finds a marked
	vertex with probability at least
	\begin{align}\label{eq:q_scenario1}
		q_\gamma&\geq\;
		p_M+(1-p_M)\bigl(\varepsilon_1 - \varepsilon_2
		- 2\varepsilon_3\sin\theta(\tilde{s}_\gamma)\bigr)^{2} \nonumber\\&\geq\;
		\bigl(\varepsilon_1 - \varepsilon_2
		- 2\varepsilon_3\sin\theta(\tilde{s}_\gamma)\bigr)^{2}\nonumber\\
		&\geq\;
		\Bigl(\tfrac{1}{2}H^{(\gamma-1)/2} - \tfrac{1}{6}H^{(\gamma-1)/2} - 2\cdot\tfrac{\sqrt2}{24}\cdot\sqrt{2}\,H^{(\gamma-1)/2}\Bigr)^{2}\nonumber\\
		&\geq\;
		\Bigl(\tfrac{1}{2} - \tfrac{1}{6} - \tfrac{1}{6}\Bigr)^{2}H^{\gamma-1}=\; \tfrac{1}{36}\,H^{\gamma-1},
	\end{align}
	where we used the upper bound of \eqref{eq:sin2_scenario1} to bound the third term. After
	$O(H^{1-\gamma})$ independent repetitions, a marked vertex is then found with
	constant probability, giving total complexity of order
	\[
	O(H^{1-\gamma})\cdot
	\bigl(\mathsf{S} + T_\gamma\cdot(\mathsf{U}+\mathsf{C})\bigr)
	\;=\;
	H^{1-\gamma}\cdot\mathsf{S}
	\;+\; H^{1-\gamma/2}\cdot(\mathsf{U}+\mathsf{C}). \qedhere
	\]
\end{proof}

\subsubsection{\texorpdfstring{Approximate $p_M$ known and $\HTp$ unknown}{Approximate pM known and HTp unknown}}

We now show that the requirement to know $\HTp$ can be relaxed by invoking
Corollary~\ref{thm:cq_approx_pm} with geometrically increasing guesses for
$\HTp$ and terminating the algorithm once a marked vertex is found.  Since each
guess determines both the interpolation parameter $\tilde{s}_\gamma$ and the
fast-forwarding time $t_j$ simultaneously, no additional structural assumption
is needed.

\begin{corollary}\label{thm:cq_approx_pm_unknown_H}
	Let $P$ and $M$ be as in Corollary~\ref{thm:cq_approx_pm}.
	Suppose an approximation $p^*$ of $p_M$ satisfying
	$|p^* - p_M| \leq p_M/3$ is given, but the value of $\HTp$ is not known.
	For every $\gamma\in[0,1]$, there exists a quantum walk-based
	algorithm with expected maximum quantum-walk depth
	\[
		D(\gamma)=O\!\left((\HTp)^{\gamma/2}\right)
	\]
	which finds a marked vertex with constant success probability and
	expected hybrid complexity of order
	\begin{equation}\label{eq:cq_approx_pm_unknown_complexity}
		\begin{cases}
			(\HTp)^{1-\gamma}\cdot \mathsf{S}
			\;+\; (\HTp)^{1-\gamma/2}\cdot(\mathsf{U}+\mathsf{C}),
			& \gamma < 1,\\[4pt]
			\log(\HTp)\cdot \mathsf{S}
			\;+\; \sqrt{\HTp}\cdot(\mathsf{U}+\mathsf{C}),
			& \gamma = 1.
		\end{cases}
	\end{equation}
\end{corollary}

\begin{proof}
	Throughout the proof we again abbreviate $H := \HTp$.  The procedure is
	given in Algorithm~\ref{alg:iis}.

	\begin{algorithm}[H]
	\caption{\textsc{InterpolatedIncrementalSearch}$(P, M, \gamma, p^*, k)$}\label{alg:iis}
	\begin{enumerate}
					\item Set $j = 0$.
					\item For the current guess $\widetilde{H} := 2^j$:
					\begin{enumerate}
						\item If
						$p^* > \tfrac{1}{2}\widetilde{H}^{\,\gamma-1}$:
						draw
						$\lceil k\,\widetilde{H}^{\,1-\gamma}\rceil$
						independent samples from $\pi$ via $\Setup(P)$ and
						check each for membership in $M$.
						\item Otherwise: set
						\begin{equation}\label{eq:sj_Tj_def}
							\tilde{s}_j
							\;:=\;
							\frac{\widetilde{H}^{\,\gamma-1} - p^*}
							{\widetilde{H}^{\,\gamma-1}(1-p^*)},
							\qquad
							t_j \;:=\; 72\,\ceil{\widetilde{H}^{\,\gamma}},
						\end{equation}
						and call
						\KMOR$(P, M, \tilde{s}_j, t_j, \varepsilon_3)$
						with $\varepsilon_3 = \sqrt2/24$, a total of
						$N_j := \lceil k\,\widetilde{H}^{\,1-\gamma}\rceil$
						times.
					\end{enumerate}
					\item If a marked vertex is found, output it and stop.
					Otherwise, set $j \leftarrow j+1$ and return to
					step~2.
				\end{enumerate}
	\end{algorithm}

	\noindent
	Here $k$ is a sufficiently large universal constant which boosts the
	success probability of each \KMOR$(P, M, \tilde{s}_j, t_j,
	\varepsilon_3)$ call, and by \eqref{eq:T-qff} the depth of such a call
	is
	\begin{equation}\label{eq:Tj_def}
		T_j \;=\; \Bigl\lceil\sqrt{2\,t_j\ln(2/\varepsilon_3)}\Bigr\rceil
		\;=\; O\bigl(\widetilde{H}^{\,\gamma/2}\bigr).
	\end{equation}

	Let
	$j_0$ denote the threshold defined by $j_0 := \lceil \log H \rceil$, so that
	$\widetilde{H}_{j_0} := 2^{j_0} \in [H, 2H)$.  We first show that every
	round $j \geq j_0$ succeeds with constant probability, then bound the
	total expected cost.

	Fix $j \geq j_0$ and consider the guess $\widetilde{H} = 2^j$.  If the classical
	branch is taken ($p^* > \tfrac{1}{2}\widetilde{H}^{\,\gamma-1}$), then
	$p_M \geq \tfrac{3}{4}p^* > \tfrac{3}{8}\widetilde{H}^{\,\gamma-1}$,
	and $N_j = \Theta(\widetilde{H}^{\,1-\gamma})$ independent samples from
	$\pi$ find a marked vertex with constant probability for some large
	enough $k$.

	Suppose instead that the quantum branch is taken, so
	$p^* \leq \tfrac{1}{2}\widetilde{H}^{\,\gamma-1}$. Whenever $\widetilde H^{\gamma-1}\ge 1/6$, we could instead use, in step 2(b), the full-depth algorithm of Theorem~7 of \cite{kmor} with
$T=\lceil 14\sqrt{\widetilde H}\rceil$, keeping the same number
$N_j$ of repetitions.
For $\widetilde H\ge H$, this gives constant
success probability, while
$\sqrt{\widetilde H}\le\sqrt{6}\,\widetilde H^{\gamma/2}$,
so the stated complexity and depth bounds are preserved.
	We may
	therefore always assume $H^{\gamma-1} < 1/6$, in which case
	$\widetilde{H}^{\,\gamma-1} < 1/6$ uniformly for every $j \geq j_0$.

	Under this assumption, the analysis of the quantum branch reduces to
	the corresponding regime in the proof of Corollary~\ref{thm:cq_approx_pm}.
	By Lemma~\ref{lem:overlap_bounds} with
	$\eta = \widetilde{H}^{\,\gamma-1}$, the same argument as in
	Corollary~\ref{thm:cq_approx_pm} shows that each call to \KMOR{} in
	step~2(b) succeeds with probability
	\[
	q_j
	\;\geq\; \bigl(\varepsilon_1 - \varepsilon_2
	- 2\varepsilon_3\sin\theta(\tilde{s}_j)\bigr)^2
	\;\geq\; \tfrac{1}{36}\,\widetilde{H}^{\,\gamma-1}
	\;=\; \Omega(\widetilde{H}^{\,\gamma-1}).
	\]
	Since
	$N_j = \Theta(\widetilde{H}^{\,1-\gamma}) = \Theta(1/q_j)$, a standard
	Chernoff bound shows that $N_j$ independent calls find a marked vertex
	with constant probability for $k$ sufficiently large.

	So in both regimes, any guess $\widetilde{H} = 2^j$ for $j\geq j_0$ finds a marked vertex with constant probability. Consequently, we fix some $k=k^*$ which guarantees there exists a constant $p_{\mathrm{fail}} \le 1/4$ such that
	any round $j \ge j_0$ fails with probability at most $p_{\mathrm{fail}}$.

	Now let $j_f$ denote the index of the final round in which a marked vertex is found.  In order to bound the
	expected total cost, we split the number of repetitions of the algorithm into
	contributions from the early rounds $j < j_0$ (which may
	all fail) and the late rounds $j \geq j_0$, each of which
	fails independently with probability at most $p_{\mathrm{fail}} \le 1/4$.

	The complexity of \IIS$(P, M, \gamma, p^*, k^*)$ is therefore
	\begin{equation}
		N^{(\mathsf S)}\cdot \mathsf S + N^{(\mathsf U)}\cdot (\mathsf U + \mathsf C),
	\end{equation}
	where $N^{(\mathsf S)}$ and $N^{(\mathsf U)}$ denote the expected total number
	of calls to $\Setup(P)$ and to the walk subroutine, that is, the expectations of
	$\sum_{j=0}^{j_f} N_j$ and $\sum_{j=0}^{j_f} N_j\,T_j$, respectively.

	We now assume every iteration for $j<j_0$ fails. Thus
	\begin{align}
		N^{(\mathsf S)}
		&\le \sum_{j=0}^{j_0-1} N_j
		+ \sum_{n=0}^{\infty} p_{\mathrm{fail}}^{\,n} N_{j_0+n} \nonumber\\
		&= O\!\left(\sum_{j=0}^{j_0-1} 2^{j(1-\gamma)}\right)
		+ O\!\left(H^{1-\gamma}\right)
		\sum_{n=0}^{\infty}
		\bigl(p_{\mathrm{fail}}\,2^{1-\gamma}\bigr)^n
		\nonumber\\  &= O\!\left(\sum_{j=0}^{j_0-1} 2^{j(1-\gamma)}\right)
		+ O\!\left(H^{1-\gamma}\right) \frac{1}{1- p_{\mathrm{fail}}2^{1-\gamma}}\nonumber\\
		&= O\!\left(\sum_{j=0}^{j_0-1} 2^{j(1-\gamma)}\right)
		+ O\!\left(H^{1-\gamma}\right),
	\end{align}
	since, by our choice of $k^*$, we have
	$p_{\mathrm{fail}}\,2^{1-\gamma} \le 1/2$.

	If $\gamma<1$, we have a geometric sum and the expectation of $j_f$ is dominated by
	\begin{align}
		N^{(\mathsf S)}
		&= O\!\left(2^{j_0(1-\gamma)}\right) = O\!\left(H^{1-\gamma}\right).
	\end{align}
	If $\gamma=1$, the first sum is $O(j_0)$ and therefore $
		N^{(\mathsf S)} = O(\log H).$
	Similarly,
	\begin{align}
		N^{(\mathsf U)}
		&\le \sum_{j=0}^{j_0-1} N_j T_j
		+ \sum_{n=0}^{\infty} p_{\mathrm{fail}}^{\,n} N_{j_0+n}T_{j_0+n} \nonumber\\
		&= O\!\left(\sum_{j=0}^{j_0-1} 2^{j(1-\gamma/2)}\right)
		+ O\!\left(H^{1-\gamma/2}\right)
		\sum_{n=0}^{\infty}
		\bigl(p_{\mathrm{fail}}\,2^{1-\gamma/2}\bigr)^n
		\nonumber\\
		&= O\!\left(2^{j_0(1-\gamma/2)}\right) + O\!\left(H^{1-\gamma/2}\right) \nonumber\\
		&= O\!\left(H^{1-\gamma/2}\right).
	\end{align}
	Therefore the expected complexity splits into two cases depending on the value of $\gamma$ as
	\begin{equation}
		N^{(\mathsf S)}\mathsf S + N^{(\mathsf U)}(\mathsf U+\mathsf C)
		=
		\begin{cases}
			H^{1-\gamma}\cdot \mathsf S
			\;+\;
			H^{1-\gamma/2}\cdot(\mathsf U+\mathsf C),
			& \gamma<1,\\[4pt]
			\log(H)\cdot \mathsf S
			\;+\;
			\sqrt H\cdot(\mathsf U+\mathsf C),
			& \gamma=1.
		\end{cases}
	\end{equation}
It remains to bound the expected maximum depth. Using the same bound on the probability
of reaching round $j_0+n$, its expectation satisfies
\begin{align}
D(\gamma)
&\le O\!\left(2^{j_0\gamma/2}\right)
 + \sum_{n=0}^{\infty} p_{\mathrm{fail}}^n
   O\!\left(2^{(j_0+n)\gamma/2}\right)\nonumber\\
&= O\!\left(H^{\gamma/2}\right)
 + O\!\left(H^{\gamma/2}\right)
   \sum_{n=0}^{\infty}
   \bigl(p_{\mathrm{fail}}2^{\gamma/2}\bigr)^n\nonumber\\
&= O\!\left(H^{\gamma/2}\right).
\end{align}
since $p_{\mathrm{fail}}2^{\gamma/2}\le \sqrt{2}/4<1$.
\end{proof}

\subsubsection{\texorpdfstring{Bound on $p_M$ known}{Bound on pM known}}

When only a lower bound $p_{\min} \leq p_M$ is available rather than an
approximation satisfying $|p^* - p_M| \leq p_M/3$, one must search over
candidate approximations.  This introduces a multiplicative overhead that
depends on whether the search is performed classically (preserving the corresponding depth constraint) or coherently via Grover's algorithm (reducing the
overhead at the cost of increased depth). We detail the former here.

\begin{corollary}\label{thm:cq_bound_pm}
	Let $P$ and $M$ be as in Corollary~\ref{thm:cq_approx_pm}.
	Suppose a lower bound $0<p_{\min}\leq p_M$ is given, with no bound on
	$\HTp$ required. For every $\gamma\in[0,1]$, there exists a quantum
	walk-based algorithm with expected maximum quantum-walk depth
	\[
		D(\gamma)=O\!\left((\HTp)^{\gamma/2}\right)
	\]
	which finds a marked vertex with constant success probability
	and expected hybrid complexity of order
	\begin{equation}\label{eq:cq_bound_pm_classical}
		\log(1/p_{\min}) \;\cdot\;
		\begin{cases}
			(\HTp)^{1-\gamma}\cdot \mathsf{S}
			\;+\; (\HTp)^{1-\gamma/2}\cdot(\mathsf{U}+\mathsf{C}),
			& \gamma < 1,\\[4pt]
			\log(\sqrt{\HTp})\cdot \mathsf{S}
			\;+\; \sqrt{\HTp}\cdot(\mathsf{U}+\mathsf{C}),
			& \gamma = 1.
		\end{cases}
	\end{equation}
	Moreover, if an upper bound $\HTp_{\max} \geq \HTp$ is also given, then
	there exists an algorithm of maximum quantum-walk depth
	$O((\HTp_{\max})^{\gamma/2})$ with expected hybrid complexity of order
	at most
	\begin{equation}\label{eq:cq_bound_pm_hmax}
		\log(1/p_{\min}) \;\cdot\;
		\Bigl(
		(\HTp_{\max})^{1-\gamma}\cdot \mathsf{S}
		\;+\; (\HTp_{\max})^{1-\gamma/2}\cdot(\mathsf{U}+\mathsf{C})
		\Bigr).
	\end{equation}
\end{corollary}

\begin{proof}
	Throughout the proof we again abbreviate $H := \HTp$.

	Let $l$ be the largest integer such that $p_M \leq 2^{-l}$.  Then
	$\tfrac{1}{2}\cdot 2^{-l} < p_M \leq 2^{-l}$, and the choice
	$p^*_l := \tfrac{2}{3}\cdot 2^{-l}$ satisfies
	\begin{equation}\label{eq:pl_approx}
		\tfrac{2}{3}\,p_M
		\;\leq\; \tfrac{2}{3}\cdot 2^{-l}
		\;=\; p^*_l
		\;=\; \tfrac{4}{3}\cdot\tfrac{1}{2}\cdot 2^{-l}
		\;<\; \tfrac{4}{3}\,p_M,
	\end{equation}
	so $|p^*_l - p_M| \leq p_M/3$.  Since $p_{\min} \leq p_M \leq 1/2$, the
	correct value of $l$ lies in $\{1, \ldots, L\}$ where
	$L := \lfloor\log(1/p_{\min})\rfloor$. Therefore, while the value of $l$ is unknown, the finite candidate set
	$\{p^*_1, \ldots, p^*_L\} = \{\tfrac{2}{3}\cdot 2^{-1}, \ldots,
	\tfrac{2}{3}\cdot 2^{-L}\}$ is guaranteed to contain a valid
	approximation, which is searched classically by Algorithm~\ref{alg:bpcs}.

\begin{algorithm}[H]
\caption{\textsc{BoundedCandidateSearch}$(P,M,\gamma,p_{\min},k)$}\label{alg:bpcs}
\begin{enumerate}
    \item Set $L:=\lfloor\log(1/p_{\min})\rfloor$ and $j:=0$.

    \item For $l'=1,\ldots,L$, set
    $p^*_{l'}:=\tfrac23\cdot2^{-l'}$ and:
    \begin{enumerate}
        \item If no upper bound $H_{\max}$ on $H$ is given,
        execute step~2 of
        \IIS$(P,M,\gamma,p^*_{l'},k)$
        with $\widetilde H:=2^j$.

        \item If an upper bound $H_{\max}\ge H$ is given,
        execute step~2 of
        \IIS$(P,M,\gamma,p^*_{l'},k)$
        with $\widetilde H:=H_{\max}$.
    \end{enumerate}
    If a marked vertex is found, output it and stop.

    \item After an unsuccessful sweep, if no upper bound
    was given, set $j\leftarrow j+1$ and return to step~2.
    Otherwise, output ``fail.''
\end{enumerate}
\end{algorithm}

	\noindent
	Note that the outer loop over $l'$ in \BPCS{} is purely classical and sequential and as such can only multiply
	the expected cost by at most a factor of $L$ but not alter the
	maximum depth of any individual inner invocation, which remains $O(2^{j\gamma/2})$.

	For the valid candidate $p_{l}^*$
		we can apply Corollary~\ref{thm:cq_approx_pm_unknown_H} which guarantees success with constant probability.
	Multiplying by the outer factor $L$
	yields~\eqref{eq:cq_bound_pm_classical}, while the expected
maximum quantum-walk depth remains $O(H^{\gamma/2})$.

	Now suppose that an upper bound $H_{\max}\ge H$ is also given. The proof of Corollary~\ref{thm:cq_approx_pm} applies entirely with the substitution $H_{\max}\leftarrow H$ and parameters
	\begin{equation}\label{eq:hmax_params}
t:=\lceil72H_{\max}^{\gamma}\rceil,
\qquad
\tilde{s}_l:=
\frac{H_{\max}^{\gamma-1}-p_l^*}
     {H_{\max}^{\gamma-1}(1-p_l^*)}.
\end{equation}
	Specifically, note that by Theorem~\ref{thm:HT-s-relation} and Lemma~\ref{lem:overlap_bounds} with
	$\eta = H_{\max}^{\gamma-1}$, we have
	\begin{equation}\label{eq:HT_hmax_bound}
		\HT(\tilde{s}_{l})
		\;=\; \sin^4\!\theta(\tilde{s}_{l})\cdot H
		\;\leq\; 4\,H_{\max}^{2(\gamma-1)}\cdot H
		\;\leq\; 4\,H_{\max}^{2\gamma-1}.
	\end{equation}
	This is the same
	upper bound on $\HT(\tilde{s}_\gamma)$ used in the proof of
	Corollary~\ref{thm:cq_approx_pm} (with $H_{\max}$ in place of $H$), so the
	fast-forwarding time \eqref{eq:hmax_params} satisfies the condition of
	Theorem~\ref{thm:main-qff} and the remainder of that proof carries through
	without modification. Therefore a marked vertex is found with constant
	probability and expected complexity
	\[
	O\bigl(
	H_{\max}^{1-\gamma}\cdot\mathsf{S}
	\;+\; H_{\max}^{1-\gamma/2}\cdot(\mathsf{U}+\mathsf{C})
	\bigr),
	\]
	and maximum depth $O(H_{\max}^{\gamma/2})$.  Again multiplying by the outer factor
	$L$ yields~\eqref{eq:cq_bound_pm_hmax}.
\end{proof}

\subsubsection{\texorpdfstring{Bound on $\HTp$ known, $p_M$ unknown}{Bound on HTp known, pM unknown}}
\label{sssec:cq_bound_H}

We treat the final variation differently from the approach in Ref.\cite{kmor}. At full depth, that setting relies on a robust dichotomic search over $p^*$ with backtracking, which requires a comparison subroutine based on eigenvalue estimation to determine whether $p_M \lessgtr p^*$. This is, however, not needed, since the only difference with Corollary~\ref{thm:cq_bound_pm} is that a lower bound $p_{\min}$ is now
unavailable, meaning the candidate sweep over the grid $\{p^*_l\}$ cannot be
truncated at a fixed level. We can instead double the length of the sweep and
halt as soon as a marked vertex is produced making the number of candidates
dependent on $p_M$ itself rather than on any a priori bound.

\begin{corollary}\label{thm:cq_bound_H}
	Let $P$ be an ergodic, reversible Markov chain on $G=(X,E)$ and let
	$M\subseteq X$ be non-empty. Suppose an upper bound
	$\HTp_{\max} \geq \HTp$ is given, but no information about $p_M$ is
	available. For every $\gamma \in [0,1]$, there is a quantum walk-based
	algorithm with maximum quantum-walk depth
	\[
		D(\gamma)=O\!\left((\HTp_{\max})^{\gamma/2}\right)
	\]
	which finds a marked vertex with constant success probability and
	expected hybrid complexity of order
	\begin{equation}\label{eq:cq_bound_H_complexity}
		\log(1/p_M)\;\cdot\;
		\Bigl(
		(\HTp_{\max})^{1-\gamma}\cdot \mathsf{S}
		\;+\; (\HTp_{\max})^{1-\gamma/2}\cdot(\mathsf{U}+\mathsf{C})
		\Bigr).
	\end{equation}
\end{corollary}

\begin{proof}
	Throughout the proof abbreviate $H := \HTp(P,M)$ and as in the previous proofs assume $p_M\le\tfrac12$, since otherwise
	a single $\Setup(P)$ draw followed by $\Checkk(M)$ returns a marked vertex
	with probability $\ge\tfrac12$.

	Following the candidate grid of Corollary~\ref{thm:cq_bound_pm}, define
	$p^*_l:=\tfrac23\cdot 2^{-l}$ and the largest integer $l^*$ with
	$p_M\le 2^{-l^*}$ satisfying $|p^*_{l^*}-p_M|\le p_M/3$ with
	\begin{equation}\label{eq:lstar}
		l^*=\lfloor\log(1/p_M)\rfloor .
	\end{equation}
	Since no lower bound on $p_M$ is available, the sweep length
	$L=\lfloor\log(1/p_{\min})\rfloor$ of \BPCS{} cannot be fixed in advance.
	It is enough to guess it. For $i=0,1,2,\ldots$ run
	\BPCS$(P, M, \gamma, \cdot, k)$ with the truncation level $L$ replaced by
	$L_i := 2^{i}$ and with the upper bound $H_{\max}$ supplied, halting as
	soon as a marked vertex is produced.  We write $i^* := \lceil \log l^*
	\rceil$ for the first attempt whose sweep is long enough to reach the
	valid candidate $p^*_{l^*}$.

	By the second part of Corollary~\ref{thm:cq_bound_pm}, a single candidate
	pass with the bound $H_{\max}$ available uses the parameters \eqref{eq:hmax_params}, and therefore costs
	\begin{equation}\label{eq:find-cost}
		\mathcal C \;:=\;
		O\bigl(H_{\max}^{1-\gamma}\mathsf S
		+ H_{\max}^{1-\gamma/2}(\mathsf U+\mathsf C)\bigr)
	\end{equation}
	at maximum quantum-walk depth $O(H_{\max}^{\gamma/2})$.  At a valid
	candidate, that is, one with $|p^*_l-p_M|\le p_M/3$, it finds a marked
	vertex with constant probability, say $p_{\text{succ}}\ge\tfrac34$ for
	some large enough $k$.

	By the same argument of Corollary~\ref{thm:cq_approx_pm_unknown_H}, for every attempt $i\ge i^*$, the algorithm fails to halt with probability
	$p_{\mathrm{fail}}\le\tfrac14$.  Attempt $i$ performs at most $2^i$ passes,
	each of cost $\mathcal C$.  Hence the expected total
	cost is
	\begin{equation}\label{eq:kmor-cost-split}
		\sum_{i<i^*}2^i\,\mathcal C
		\;+\;
		\sum_{n\ge 0} p_{\mathrm{fail}}^{\,n}\,2^{\,i^*+n}\,\mathcal C
		\;\le\;
		2^{i^*}\mathcal C
		\;+\; 2^{i^*}\mathcal C\sum_{n\ge0}\bigl(2p_{\mathrm{fail}}\bigr)^n
		\;\le\; 3\cdot 2^{i^*}\,\mathcal C\,,
	\end{equation}
	and since $2^{i^*}<2l^*$, the expected
	cost is
	\[
	O(l^*)\,\mathcal C
	\;=\; O(\log(1/p_M))\cdot
	\bigl(H_{\max}^{1-\gamma}\mathsf S
	+ H_{\max}^{1-\gamma/2}(\mathsf U+\mathsf C)\bigr),
	\]
	while maintaining the maximum quantum-walk depth $O(H_{\max}^{\gamma/2})$
	of a single pass.
\end{proof}

\subsubsection{Interpolated detection}
\label{ssec:detection}

The one caveat of not having a lower bound on $p_M$ is that,
when it is unknown whether the marked set is empty or not, the candidate
sweep may never halt. Nonetheless, given an upper bound on the
hitting time, we may first detect the
existence of marked vertices with at most constant overhead, by the following reduction.

\begin{proposition}[\cite{Sze04a,magniez2012hitting}]
\label{prop:detection-reduction}
Let $P$ be an ergodic, reversible Markov chain and let
$M\subseteq X$. Suppose an upper bound $\HT\geq1$ is given with the promise
that either $M=\emptyset$ or $\HT(P,M)\leq\HT$.

The problem of deciding whether $M$ is non-empty with constant probability reduces to eigenphase estimation
of $W(P')$ on the unmarked state $\ket{\pi_U}$, with precision
$\epsilon=\Theta(\HT^{-1/2})$. In particular, an empty marked set is always reported empty.
\end{proposition}

Note that no interpolation parameter $s$ is needed since the walk operator is that of the absorbing chain $P'$. This means the cost is dominated by the ordinary hitting time $\HT(P,M)$ and no knowledge of $p_M$ is needed.

The idea is that when $M = \emptyset$, the state $\ket{\pi_U}$ is an invariant eigenvector of $W(P')$ with eigenphase zero, returning the measurement outcome $0^t$ with certainty. In contrast, if $M$ is non-empty, a constant fraction of the spectral weight possesses non-zero eigenphases of order $\Omega(\HT^{-1/2})$. Rather than resolving these phases with $O(1)$ circuits of depth $\Theta(\HT^{1/2})$, we may instead detect these shifts with reduced precision using the interpolated phase estimation results of Refs.~\cite{wang2019accelerated,giurgica2022qjournal,Rall2023amplitudeestimation}.

\begin{corollary}\label{cor:interpolated-detection}
	Let $P$ be an ergodic, reversible Markov chain and let
	$M\subseteq X$. Suppose an upper bound $\HT\geq1$ is given with the
	promise that either $M=\emptyset$ or $\HT(P,M)\leq\HT$.
	For every $\gamma\in[0,1]$, there is a quantum walk-based
	algorithm of quantum-walk depth
	\[
		D(\gamma)=O\!\left(\HT^{\gamma/2}\right)
	\]
	that decides whether $M$ is non-empty with constant success
	probability and hybrid complexity
	\begin{equation}\label{eq:interp-detect-cost}
		O\!\left(
		\HT^{1-\gamma}\cdot\mathsf{S}
		\;+\;
		\HT^{1-\gamma/2}\cdot(\mathsf{U}+\mathsf{C})
		\right).
	\end{equation}
\end{corollary}

\begin{proof}
	By Proposition~\ref{prop:detection-reduction}, it is enough to resolve a phase of order $\epsilon:=\Omega(\HT^{-1/2})$, following state preparation at cost $\mathsf S + \mathsf C$. Using interpolated phase estimation, this can be achieved in time $O(\epsilon^{\gamma-2})$ and quantum walk depth $D(\gamma)=O(\epsilon^{-\gamma})$, with each walk application costing $O(\mathsf{U}+\mathsf{C})$, across $O(\epsilon^{2(\gamma-1)})$ repetitions.
\end{proof}

Thus, if $|M|=1$, detecting before finding contributes at most
a constant factor and is always cheaper for the multiple marked vertices case.

\subsection{AGJK with unknown hitting time}

For the AGJK algorithm, no a priori knowledge of $p_M$ is required. We may only be limited by the lack of an upper bound on $\HT(P,M)$, which can be resolved with the
standard doubling strategy used previously.

\begin{corollary}\label{cor:interp-qff-unknown}
	Let $P$ be an ergodic, reversible Markov chain on $G=(X,E)$, and let
	$M\subseteq X$ be non-empty, with no bound on $\HT:=\HT(P,M)$ given.
	For every $\gamma\in[0,1]$, there is a quantum walk-based
	algorithm with expected maximum quantum-walk depth
	\begin{equation}\label{eq:interp-qff-unknown-depth}
		D(\gamma)=O\!\left(
			\HT^{\gamma/2}
			\sqrt{(1-\gamma)\log\HT+\log\log\HT}
		\right)
	\end{equation}
	which outputs a marked vertex with constant success probability
	and expected hybrid complexity of order at most
	\begin{equation}\label{eq:interp-qff-unknown-complexity}
		\begin{cases}
			\begin{aligned}
				&(1+\gamma\log\HT)\,
				\HT^{1-\gamma}\cdot\mathsf S\\[-1pt]
				&\quad{}+
				(1+\gamma\log\HT)\,
				\HT^{1-\gamma/2}
				\sqrt{(1-\gamma)\log\HT+\log\log\HT}
				\cdot(\mathsf U+\mathsf C)
			\end{aligned},
			& \gamma<1,\\[6pt]
			\log^2(\HT)\cdot\mathsf S
			\;+\;
			\sqrt{\HT}\,\log(\HT)\sqrt{\log\log\HT}
			\cdot(\mathsf U+\mathsf C),
			& \gamma=1.
		\end{cases}
	\end{equation}
\end{corollary}

\begin{proof}
	We repeatedly invoke
	Theorem~\ref{thm:interp-qff} using exponentially increasing guesses for the
	unknown hitting time until a marked vertex is found, as in Algorithm~\ref{alg:iiagjk}.

	\begin{algorithm}[H]
	\caption{\textsc{IncrementalAGJK}$(P, M, \gamma)$}\label{alg:iiagjk}
	\begin{enumerate}
					\item Set $j = 0$.
					\item For the current guess $\widetilde{H} := 2^j$ call \IAGJK$(P, M, \widetilde{H}, \gamma)$ a total of $R(\gamma)= O\left(\left(1+\gamma\log \widetilde H\right)\widetilde H^{1-\gamma}\right)$ times.
					\item If a marked vertex is found, output it and stop.
					Otherwise, set $j \leftarrow j+1$ and return to step~2.
				\end{enumerate}
	\end{algorithm}

	\noindent

	Let $j_0 := \lceil\log(3\HT)\rceil$ be the index where the current guess on the hitting time bound takes the value
	$\widetilde{H}_{j_0} = 2^{j_0} \in [3\HT, 6\HT)$. Then, for every $j \ge j_0$, by Theorem~\ref{thm:interp-qff}, \IAGJK$(P, M, \widetilde{H},\gamma)$ finds a marked vertex with constant probability, say
	$1 - p_{\mathrm{fail}} \ge \tfrac34$ for some large enough multiplicative constant in the number of repetitions.

	Additionally, let $j_f$ denote the final round. Conditioned on reaching
	round $j$, the depth of that round is, by
	\eqref{eq:interp-qff-depth},
	\begin{align}
	D^{(j)}(\gamma) &= O\!\left(\widetilde{H}^{\,\gamma/2}
	\sqrt{(1-\gamma)\log\widetilde H + \log\log\widetilde H}\right)
	\nonumber\\ &= O\!\left(2^{j\gamma/2}
	\sqrt{(1-\gamma)\,j + \log j}\right),
	\end{align}
	which is increasing in $j$. Since every round $j \ge j_0$ succeeds with
	probability at least $\tfrac34$, the index $j_f$ exceeds $j_0 + n$ with
	probability at most $p_{\mathrm{fail}}^{\,n} \le 4^{-n}$, so with constant
	probability $j_f \le j_0 + O(1)$ and the maximum depth encountered is
	\begin{equation}
D(\gamma) = O\bigl(\HT^{\gamma/2}\sqrt{(1-\gamma)\log\HT +
		\log\log\HT}\bigr).
\end{equation}
	 By Theorem~\ref{thm:interp-qff}, that same call at round $j$ costs
	\begin{equation}\label{eq:round-cost}
		\mathcal C_j
		\;=\; O\!\left(
		(1+\gamma j)\,2^{j(1-\gamma)}\cdot\mathsf S
		\;+\;
		(1+\gamma j)\,2^{j(1-\gamma/2)}\sqrt{(1-\gamma)\,j + \log j}
		\cdot(\mathsf U+\mathsf C)\right).
	\end{equation}
	Now, we (pessimistically) assume that every round $j < j_0$ fails to return a marked vertex, while each round $j \ge j_0$ fails independently with probability at
	most $p_{\mathrm{fail}} \le \tfrac14$. The expected total cost is then
	bounded by
	\begin{align}\label{eq:cost-split}
		&\;\;\sum_{j=0}^{j_0 - 1}\mathcal C_j
		\;+\;
		\sum_{n \ge 0} p_{\mathrm{fail}}^{\,n}\,\mathcal C_{j_0 + n}\nonumber\\
		&=N^{(\mathsf{S})}\cdot \mathsf{S}+N^{(\mathsf{U})}\cdot(\mathsf{U}+\mathsf{C}),
	\end{align}
	where $N^{(\mathsf{S})}$ and $N^{(\mathsf{U})}$ denote the expectation of the number of calls to $\Setup(P)$ and to the walk unitary, respectively.
	Explicitly, $N^{(\mathsf{S})}$ is given by
	\begin{align}
		N^{(\mathsf S)}
		&=
		O\!\left(
		\sum_{j=0}^{j_0-1}
		(1+\gamma j)\,2^{j(1-\gamma)}
		+
		\sum_{n\ge0}
		p_{\mathrm{fail}}^{\,n}
		(1+\gamma(j_0+n))\,2^{(j_0+n)(1-\gamma)}
		\right).
	\end{align}
	For $\gamma<1$, the factor $2^{j(1-\gamma)}$ grows geometrically, so the first sum is dominated by its final term
	\begin{align}
		\sum_{j=0}^{j_0-1}
		(1+\gamma j)\,2^{j(1-\gamma)}
		&=
		O\!\left(
		(1+\gamma j_0)\,
		2^{j_0(1-\gamma)}
		\right)\nonumber\\
		&=
		O\!\left(
		(1+\gamma\log\HT)\,
		\HT^{1-\gamma}
		\right),
	\end{align}
	where we used that $j_0=\Theta(\log\HT)$ and
	$2^{j_0}\in[3\HT,6\HT)$. Similarly,
	\begin{align}
		\sum_{n\ge0}
		p_{\mathrm{fail}}^{\,n}
		(1+\gamma(j_0+n))\,2^{(j_0+n)(1-\gamma)}
		&=
		2^{j_0(1-\gamma)}
		\sum_{n\ge0}
		\bigl(p_{\mathrm{fail}}2^{1-\gamma}\bigr)^n
		(1+\gamma(j_0+n))\nonumber\\
		&=
		O\!\left(
		(1+\gamma j_0)\,
		2^{j_0(1-\gamma)}
		\right)\nonumber\\
		&=
		O\!\left(
		(1+\gamma\log\HT)\,
		\HT^{1-\gamma}
		\right),
	\end{align}
	since $
	p_{\mathrm{fail}}\,2^{1-\gamma}
	\le
	2^{-1-\gamma}
	<1$.

	Hence, for every $\gamma<1$,
	\begin{equation}
	N^{(\mathsf S)}
	=
	O\!\left(
	(1+\gamma\log\HT)\,
	\HT^{1-\gamma}
	\right).
	\end{equation}
	For $\gamma=1$, the geometric factor disappears and we instead obtain
	\begin{align}
		N^{(\mathsf S)}
		=
		O\!\left(
		\sum_{j=0}^{j_0-1}(1+j)
		+
		\sum_{n\ge0}
		p_{\mathrm{fail}}^{\,n}(1+j_0+n)
		\right)
		=
		O\!\left(
		j_0^2+j_0
		\right)
		=
		O\!\left(
		\log^2\HT
		\right).
	\end{align}
	We now turn to the contribution of the update and checking costs, which is given by
	\begin{align}
		N^{(\mathsf U)}
		&\le
		\sum_{j=0}^{j_0-1}
		(1+\gamma j)\,
		2^{j(1-\gamma/2)}
		\sqrt{(1-\gamma)j+\log j}\nonumber\\
		&+
		\sum_{n\ge0}
		p_{\mathrm{fail}}^{\,n}
		(1+\gamma(j_0+n))\,
		2^{(j_0+n)(1-\gamma/2)}
		\sqrt{(1-\gamma)(j_0+n)+\log(j_0+n)}.
	\end{align}
Since $1-\gamma/2\ge\tfrac12$, the first sum is again dominated by its final term
	\begin{align}
		&\sum_{j=0}^{j_0-1}
		(1+\gamma j)\,
		2^{j(1-\gamma/2)}
		\sqrt{(1-\gamma)j+\log j}\nonumber\\
		&=
		O\!\left(
		(1+\gamma j_0)\,
		2^{j_0(1-\gamma/2)}
		\sqrt{(1-\gamma)j_0+\log j_0}
		\right)\nonumber\\
		&=
		O\!\left(
		(1+\gamma\log\HT)\,
		\HT^{1-\gamma/2}
		\sqrt{(1-\gamma)\log\HT+\log\log\HT}
		\right).
	\end{align}
Likewise,
	\begin{align}
		&\sum_{n\ge0}
		p_{\mathrm{fail}}^{\,n}
		(1+\gamma(j_0+n))\,
		2^{(j_0+n)(1-\gamma/2)}
		\sqrt{(1-\gamma)(j_0+n)+\log(j_0+n)}\nonumber\\
		&=
		2^{j_0(1-\gamma/2)}
		\sum_{n\ge0}
		\bigl(p_{\mathrm{fail}}2^{1-\gamma/2}\bigr)^n
		(1+\gamma(j_0+n))
		\sqrt{(1-\gamma)(j_0+n)+\log(j_0+n)}\nonumber\\
		&=
		O\!\left(
		(1+\gamma j_0)\,
		2^{j_0(1-\gamma/2)}
		\sqrt{(1-\gamma)j_0+\log j_0}
		\right)\nonumber\\
		&=
		O\!\left(
		(1+\gamma\log\HT)\,
		\HT^{1-\gamma/2}
		\sqrt{(1-\gamma)\log\HT+\log\log\HT}
		\right),
	\end{align}
using $
	p_{\mathrm{fail}}\,2^{1-\gamma/2}
	\le
	2^{-1-\gamma/2}
	<1$.

	Therefore we recover the original $(\mathsf U+\mathsf C)$ term for all values of $\gamma$
	\begin{equation}
	N^{(\mathsf U)}
	=
	O\!\left(
	(1+\gamma\log\HT)\,
	\HT^{1-\gamma/2}
	\sqrt{(1-\gamma)\log\HT+\log\log\HT}
	\right).
	\end{equation}
	Substituting the bounds on $N^{(\mathsf S)}$ and $N^{(\mathsf U)}$ into
	\eqref{eq:cost-split} yields \eqref{eq:interp-qff-unknown-complexity}.
\end{proof}

\end{document}